\documentclass[11pt,a4paper,twoside,dvips]{article}

\usepackage[a4paper,hmargin=25mm,vmargin=25mm]{geometry}%

\usepackage{graphicx}%
\usepackage[usenames]{xcolor}

\usepackage{pstricks}
\usepackage{psfrag}

\usepackage{tikz}
\usetikzlibrary{calc,intersections,through,backgrounds}

\usepackage{enumerate}

\usepackage[nodayofweek]{datetime} 
\newdateformat{mydate}{ \monthname[\THEMONTH], \THEYEAR}

\usepackage[colorlinks,breaklinks=true]{hyperref}

\usepackage{scrextend}
\deffootnotemark{\textsuperscript{*\thefootnotemark}}
\deffootnote{2em}{1.6em}{*\thefootnotemark\enskip}

\usepackage[reqno]{amsmath}
\usepackage{amsthm}
\usepackage{amsfonts}
\usepackage{amssymb}
\usepackage{bm}

\numberwithin{equation}{section}

\newcommand{\rmd}{\text{d}}
\newcommand{\del}{\partial}
\newcommand{\lieL}{\mathcal{L}}
\newcommand{\symp}[2]{\omega \left( #1 , #2 \right)}
\newcommand{\order}[1]{\mathcal{O} \left( #1 \right)}

\theoremstyle{plain}
\newtheorem{thm}{Theorem}[section]
\newtheorem{cor}[thm]{Corollary}

\newtheorem{conj}[thm]{Conjecture}

\theoremstyle{remark}
\newtheorem{rmk}{Remark}[section]
\newtheorem{nota}[rmk]{Note}

\theoremstyle{definition}
\newtheorem*{defn}{Definition}

\theoremstyle{plain}

\newtheorem*{morseLem}{Morse Lemma}

\usepackage{datetime}
\usepackage[displaymath,pagewise]{lineno}

\begin{document}


\begin{center} 
{\huge Bifurcations of transition states: Morse bifurcations}\\[10pt]

{\large R S MacKay$^{1 \; 2\; \dagger}$ and D C Strub$^{1 \; *}$}\\[10pt]

$^1$ Mathematics Institute and $^2$ Centre for Complexity Science, University of Warwick, Coventry CV4 7AL, U.K.

$^\dagger$ R.S.MacKay@warwick.ac.uk, $^*$ D.C.Strub@warwick.ac.uk\\[30pt]

Published in {\bf Nonlinearity}, \href{http://dx.doi.org/10.1088/0951-7715/27/5/859}{24(5):859-895} (Open Access)
\\[30pt]
\end{center}

\begin{abstract}
A transition state for a Hamiltonian system is a closed, invariant, oriented, codimension-2 submanifold of an energy-level that can be spanned by two compact codimension-1 surfaces of unidirectional flux whose union, called a dividing surface, locally separates the energy-level into two components and has no local recrossings.
For this to happen robustly to all smooth perturbations, the transition state must be normally hyperbolic.
The dividing surface then has locally minimal geometric flux through it, giving an upper bound on the rate of transport in either direction.

Transition states diffeomorphic to $\mathbb S^{2m-3}$ are known to exist for energies just above any index-1 critical point of a Hamiltonian of $m$ degrees of freedom, with dividing surfaces $\mathbb S^{2m-2}$. The question addressed here is what qualitative changes in the transition state, and consequently the dividing surface, may occur as the energy or other parameters are varied? 
We find that there is a class of systems for which the transition state becomes singular and then regains normal hyperbolicity with a change in diffeomorphism class. These are Morse bifurcations.

Various examples are considered. Firstly, some simple examples in which transition states connect or disconnect, and the dividing surface may become a torus or other. Then, we show how sequences of Morse bifurcations producing various interesting forms of transition state and dividing surface are present in reacting systems, by considering a hypothetical class of bimolecular reactions in gas phase. 
\end{abstract}




\newpage

\section{Introduction}

Many physical problems can be reduced to considering the rate of transport of volume between different regions of the state space\footnote{The state space of Hamiltonian systems is often referred to as phase space.} of a Hamiltonian system. Such problems arise when studying chemical reaction rates \cite{Keck1967}, capture and escape processes in celestial mechanics \cite{jaffe2002statistical},  particle escape from charged particle storage rings \cite{Pinheiro2008}, and displacement of defects in solids \cite{Toller1985}. 
The state space volume of interest in these problems is that occupied by a classical ensemble of a low-dimensional dynamical system derived from a ``kinetic approximation'', i.e. a great number of different, independent realisations of the same Hamiltonian system.
Since energy is conserved by the systems, it is natural to consider the rate at which energy-surface volume crosses between regions as a function of the energy. 

Both Marcelin \cite[Ch.2]{Marcelin} and later Wigner \cite{wigner1937calculation} realised that a natural way to study the rate of transport is to place a \textit{dividing surface}\footnote{Nb: `surface' means an (embedded) submanifold of codimension-1.} between the regions of interest and consider the flux of energy-surface volume through this surface. It must divide the whole energy level into two distinct regions, such that no trajectory can pass from one region to the other without crossing it. In this case the flux of energy-surface volume through the dividing surface in one direction gives an upper bound on the rate of transport in this direction. In order to obtain a useful upper bound, Wigner \cite{wigner1937calculation} proposed varying the dividing surface to obtain one with (locally) minimal flux. This variational definition does not determine a unique dividing surface, because one can flow any dividing surface along the vector field and obtain another, but the minimum flux is well defined.

A lot of the initial and subsequent work was done with reaction rates in mind and led to variational transition state theory (nicely reviewed and put into context by Pollak and Talkner \cite{pollak2005reaction}). In chemical reactions, one finds the basic transport scenario, which we refer to as ``flux over a saddle''. Here, the two regions of interest represent reactants and products and have an index-1 critical point of the Hamiltonian function between them, hence the name.
For small energies above that at the critical point, the energy level is diffeomorphic to a spherical cylinder and has a bottleneck about the critical point. The dividing surface is therefore placed in the bottleneck region, thus allowing for a local analysis. It can be decomposed into two parts with oppositely directed flows, each of which spans\footnote{We say that a manifold $S$ spans $N$ if the latter is its boundary, $N = \del S$.} a normally hyperbolic\footnote{An invariant submanifold whose linearised normal dynamics are hyperbolic and dominate those tangent to it, see Appendix~\ref{nhms}.}, closed, codimension-2 submanifold, the \textit{transition state}. 
The general picture was given by Wigner \cite{Wigner1938}. Explicit dividing surfaces spanning a family of unstable periodic orbits were constructed by Pechukas in the case of 2 degrees of freedom \cite{Pechukas1976}. Toller et al.~\cite{Toller1985} generalised this to arbitrary degrees of freedom, without however restricting to energy-levels. These were then considered by MacKay \cite{MacKay1990}, though for more than 2 degrees of freedom it left too much to the imagination of the reader, which we rectify here. The construction was rediscovered by Wiggins et al. \cite{wiggins2001impenetrable}, but without explanation of how to pass from the normal form approximations to the full problem.

We are interested in what happens when the energy is increased in the basic scenario and the bottleneck opens up, and also in more general transport problems. 
These scenarios will not allow for a local analysis about some critical point, and will involve global dividing surfaces and transition states that may bifurcate leading to other scenarios.
For two degree of freedom systems, the transition states are hyperbolic periodic orbits. The bifurcations of such objects are well known and can be found in the literature, e.g. Meyer et al. \cite[Ch.11]{meyer2009introduction} and Han\ss mann \cite[Ch.3]{hanssmann2007local}.
Thus, there have been various studies of the bifurcations of periodic orbit transition states in the transition state theory literature, see e.g.~\cite{Pollak1978,Davis1987}. 
For higher degrees of freedom, the transition states are normally hyperbolic submanifolds of higher dimensions, e.g.~in the basic scenario we have a $\left(2m-3\right)$-sphere. The bifurcation of such submanifolds is not much explored or understood, even though there have been a few recent studies \cite{li2009bifurcation,Teramoto2011,Allahem2012}. It has been considered a hard problem because bifurcation necessitates loss of normal hyperbolicity and for submanifolds of dimension 3 or greater there are many possible consequences.
However, what has been overlooked and will be considered here, is that there is a large class of systems for which normal hyperbolicity is regained immediately: the transition state develops singularities at a critical energy, i.e. points at which the manifold structure fails, 
but regains smoothness and normal hyperbolicity with a change in diffeomorphism type\footnote{This possibility was pointed out by MacKay in the course of a workshop in Bristol in 2009. Some examples have been reported in Maugui\`ere et al. \cite{Mauguiere2013}, a paper that appeared in preprint form at about the same time as ours.}.
The dividing surfaces also undergo a similar bifurcation. 
These are \textit{Morse bifurcations}. 

Section \ref{review} is an introduction and review of the dividing surface method for transport problems. After giving the basic definitions, we consider surfaces within an energy level that have locally minimal flux of energy-surface volume across them. This aims to be a (fairly) complete review of the fragmented literature found partly in chemistry and partly in dynamical systems. We follow the introduction to flux of energy-surface volume by MacKay \cite{MacKay1990}. Our main sources for the surfaces of locally minimal flux are Wigner \cite{wigner1937calculation}, Keck \cite{Keck1967}, Toller et al.~\cite{Toller1985} and MacKay \cite{mackay1991variational, MacKay1994}. Next, we consider the basic transport scenario, that of flux over a saddle. 
The dividing surface method was thought up with this scenario in mind, so here we find all the basic ingredients.
In Section \ref{transition}, we give the general definitions and outline the properties of the submanifolds that are needed for the dividing surface method. These generalise the ideas found in the basic scenario.
Next, Section \ref{variation} gives a neat method to construct dividing surfaces about transition states, which uses the normal hyperbolicity of the latter. This is a necessary step in order to consider more general transport problems.
Readers who are already well versed in the subject may skip ahead to Section \ref{bifurcations} which is where the main contribution of this paper starts. This contains a discussion on the possible bifurcations of the dividing surfaces and transition states, and then gives some examples of Morse bifurcations. The examples initially have 2 degrees of freedom, such that the Morse bifurcations stand out.
However, Morse theory, which is briefly reviewed in Appendix \ref{morse}, applies to manifolds of all dimensions, so there is no limitation to the number of degrees of freedom one can consider.
The central role played by Morse bifurcations in transport problems is seen in the application given in Section \ref{bimolecular}. Here, we find interesting sequences of Morse bifurcations and therefore transition states and dividing surfaces, in prototype planar bimolecular reactions.
In Section \ref{conclude}, we conclude our study of Morse bifurcations and discuss the remaining open questions for the dividing surface method.
Following are three appendices, one on how the flux as a function of energy changes at Morse bifurcations, another on approximating normally hyperbolic (symplectic) submanifolds and finally one on Morse theory.


\section{Review}
\label{review}

In order to consider the flow of state space volume between different regions of a Hamiltonian system, we start with some definitions and terminology in Section \ref{divSurfs}. Then, we consider codimension-1 submanifolds of an energy level with locally minimal flux through them, and finally we review the basic transport scenario of flux over a saddle.


\subsection{Flux of energy-surface volume across a surface}
\label{divSurfs}

We consider autonomous Hamiltonian systems with $m$ degrees of freedom $\left( M^{2m}, \omega, H \right)$. Here $M$ is a $2m$-dimensional differentiable manifold, $\omega$ a symplectic form\footnote{A closed, non-degenerate 2-form.} on it and $H$ a smooth function from $M$ to $\mathbb{R}$.

The Hamiltonian vector field $X_H$ is defined by 
$$i_{X_H} \omega = \rmd H,$$
where the interior product $i$ contracts a vector field $X$ and a $k$-form $\alpha$ to give a $\left( k-1 \right)$-form, by
$$ i_X \alpha \left( \nu_1, \cdots , \nu_{k-1} \right) = \alpha \left( X, \nu_1, \cdots , \nu_{k-1} \right),$$
for any $\left( k-1 \right)$ vectors $\nu_i$.

We are interested in the flow of state space (or Liouville) volume
$$\Omega = \frac{1}{m!}\omega^m.$$
In local Darboux coordinates $z = \left(q,p \right)$, $\omega = \sum_{j=1}^{m} \rmd q_j \wedge \rmd p_j$ and $\Omega = \rmd q_1 \wedge \rmd p_1 \wedge \cdots \wedge \rmd q_m \wedge \rmd p_m$.

Since $H$ is conserved by the flow, we can restrict attention to the energy levels $M_E = H^{-1} \left( E \right)$ and so consider energy-surface volume, $\Omega_E$, a $\left(2m-1 \right)$-form on regular $M_E$ (i.e.~$\rmd H$ is nowhere zero on $M_E$; $E \in \mathbb{R}$ is said to be a regular value of $H$), given by the relation 
$$\rmd H \wedge \Omega_E = \Omega.$$

We can now define the (energy-surface) \textit{flux form}
$$\phi_E = i_{X_{H}}\Omega_E.$$
At times, we will also need to consider the state space flux form simply given by
$$\phi = i_{X_H} \Omega .$$
Finally, the flux of energy-surface volume through a codimension-1, oriented submanifold $S_E$ of an energy level is the integral 
$$\phi_E \left( S_E \right) = \int_{S_E} \phi_E.$$
\begin{rmk}
We have chosen to use differential forms, but the equation is the same as the usual flux equation $\phi_E \left( S_E \right) = \int_{S_E} \left( X_H \cdot n \right) \text{vol}_{S_E} $, where $n$ is the unit normal to the surface (with respect to a Riemannian structure) and $\text{vol}_{S_E} $ an infinitesimal volume element, cf. Keck \cite{Keck1967}. This is seen by rewriting the flux form as $\phi_E = i_{X_{H}}\Omega_E = \left( X_H \cdot n \right) i_{n}\Omega_E$ and defining $\text{vol}_{S_E} := i_{n}\Omega_E$. See Frankel \cite[\S2.9b]{Frankel2004} for more details. Use of a Riemannian structure however is an artificial crutch.
\end{rmk}

Evaluating this integral can be simplified by noting that for regular energy levels $M_E$, the flux form reduces to $\phi_E = \omega^{m-1} / \left( m - 1 \right) !$ \cite{Toller1985}.
The symplectic form $\omega$ can, modulo global topological obstructions, be written as $\omega = - \rmd \lambda$, for some 1-form $\lambda$, called the \textit{action form}. In local Darboux coordinates one can take $\lambda = \sum_{k=1}^{m} p_k \rmd q_k $.
This allows us to write 
$$\phi_E = - \rmd \Lambda,$$
with the \textit{``generalised'' action form} $\Lambda = \frac{1}{\left(m-1\right)!} \lambda \wedge \omega^{m-2}$.
Then, we can use Stokes' theorem to obtain
\begin{equation*}
\phi_E \left( S_E \right) = \int_{S_E} \phi_E = - \int_{\partial S_E} \Lambda.
\end{equation*}
We state this as
\begin{thm}[\cite{MacKay1990}]
\label{fluxStoke}
The flux of energy-surface volume through an oriented codimension-1 submanifold
of an energy-level is minus the generalised action integral over its boundary.
\end{thm}
In general the flux form $\phi_E$ evaluated on the tangent space to an oriented surface is not single-signed. The flow may cross in the positive direction in some parts of $S_E$ and the other way on other parts. Where we want to emphasise this we refer to the flux integral as ``net flux''.


\subsection{Surfaces of locally minimal flux}
\label{sectMin}

Consider trajectories crossing from a region $L_E$ to another $R_E$, within an energy level $M_E$, and imagine a dividing surface somewhere between $L_E$ and $R_E$, which is crossed by all these trajectories. The positive part of the flux of energy-surface volume through this surface gives an upper bound on the rate of transport.
In order to obtain a useful bound on the transport, we vary the arbitrary dividing surface and replace it with a dividing surface with (close to) locally minimal flux in the chosen direction. 
Ideally, we would like to find a dividing surface that is crossed once and only once by each trajectory crossing from $L_E$ to $R_E$, which would have minimal flux through it and would give the actual rate of transport.
Hence, we have a variational definition of the (ideal) dividing surface.
However, this condition does not necessarily define a unique dividing surface because one can flow any dividing surface along the vector field and obtain another. Furthermore the minimal dividing surfaces might be hard to find in practice.

There is also a flow in the opposite direction, from $R_E$ to $L_E$, for which we could choose to consider the previous dividing surface, extended to cut all trajectories from $R_E$ to $L_E$.
In order to divide the energy-level and therefore separate the two regions, we expect the surface to be closed (i.e. without boundary), otherwise trajectories could avoid it but still cross between regions.
Thus, the net flux through the dividing surface will be zero, by Theorem \ref{fluxStoke}, and the flux in the two separate directions equal.

We will now consider which surfaces have locally minimal flux, in either direction, and the properties of such surfaces.
We follow the approach of MacKay \cite{MacKay1994}, which lends itself well to the general scenarios that we want to consider.

Invariant surfaces have zero flux in both directions. This is clearly a minimum value and when closed these surfaces are useful for transport problems because they separate different regions.
We are interested however in surfaces that have locally minimal but not necessarily zero flux in each direction.

For a general closed, oriented surface $S_E$, one can decompose it as the union $S_E^+ \cup S_E^-$ of the parts of positive and negative flux\footnote{By ``positive'' we mean ``non-negative'', and by ``negative'' we mean ``non-positive'' but the terminology is too cumbersome.}, with common boundary. Then we can apply the following
consequence of the variational principle for odd dimensional invariant submanifolds of an energy level.
\begin{cor}[\cite{mackay1991variational}]
\label{statFlux}
A  codimension-1 submanifold $S_E^i$
of an energy level $M_E$ has stationary (net) flux of energy-surface volume with respect to variations including of its boundary $\partial S_E^i$ within $M_E$ if and only if $\partial S_E^i$ is invariant under the Hamiltonian flow.
\end{cor}

These stationary values of the flux are however not minima, since
there exist deformations that both increase and decrease the flux.
A nice way of seeing this is to consider deforming $\del S_E^i$ to a helix of small pitch around a part of $\del S_E^i$ with $X_H \neq 0$. Remembering that $X_H$ must be tangent to $\del S_E^i$ for stationary values, we find that the flux increases if the pitch has one sign and decreases if it has the other \cite{MacKay1994}.

We therefore consider (unsigned) \textit{geometric flux} of energy-surface volume though $S_E$, denoted $\Phi_E \left( S_E \right) $. For this, we define the (energy-surface) \textit{flux density} $\lvert \phi_E \rvert$, by
\begin{equation*}
\vert \phi_E \vert \left( \nu_1, ..., \nu_{2m-2} \right) = \vert \phi_E \left( \nu_1, ..., \nu_{2m-2} \right) \vert, \quad \forall \nu_i \in T_z M_E,
\end{equation*}
and integrate it over $S_E$. For a brief introduction to densities and density bundles see e.g. Lee \cite[Ch.14]{lee2003smooth}.

For an arbitrary dividing surface, $\Phi_E \left( S_E \right) \geq \lvert \phi_E \left( S_E \right) \rvert$, with equality if and only if the flux is unidirectional through the surface.
By unidirectional, we mean single-signed, as this occurs when the Hamiltonian vector field $X_H$ is unidirectional across $S_E$.
Thus, decomposing our closed dividing surface as $S_E = S_E^+ \cup S_E^-$, where $S_E^\pm$ are not necessarily connected, gives 
$$\Phi_E \left(S_E \right) = \phi_E \left(S_E^+ \right) - \phi_E \left(S_E^- \right) = 2 \phi_E \left(S_E^+ \right) = -2 \phi_E (S_E^-),$$
since the flux is equal and opposite through $S_E^+$ and $S_E^-$.
Asking for minimal flux in either direction is therefore the same as asking for minimal geometric flux. Furthermore, closed surfaces either have both stationary net flux and stationary geometric flux or neither.

The idea is therefore to consider
dividing surfaces $S_E$ that are the union of two, not necessarily connected, surfaces $S_E^\pm$ of unidirectional flux that span a closed, invariant, codimension-2 orientable submanifold $N_E$ in the energy level, $M_E$.
Then $S_E^\pm$ have stationary flux, by Corollary \ref{statFlux}, and the geometric flux is
$$ \Phi_E \left(S_E \right) = - 2 \phi_E \left( S_E^- \right) =- 2 \int_{S_E^-} \phi_E = 2 \int_{N_E} \Lambda ,$$
where we have chosen the orientation of $N_E$ so that it is $\del S_E^-$.
To show that this situation leads to locally minimal geometric flux, we first need the following 
\begin{defn}
An oriented surface $S$ has \textit{local recrossings} if for all $\varepsilon > 0$ there exists an orbit segment $z(t)$, $t_0 \leq t \leq t_1$, that intersects $S$ in opposite directions at times $t_0$ and $t_1$, and for which 
$$0<d(z(t),S)<\varepsilon \;\; \text{  for all  } \;\; t \in \left( t_0 , t_1 \right),$$
where $d$ denotes distance.
\end{defn}
\begin{figure}
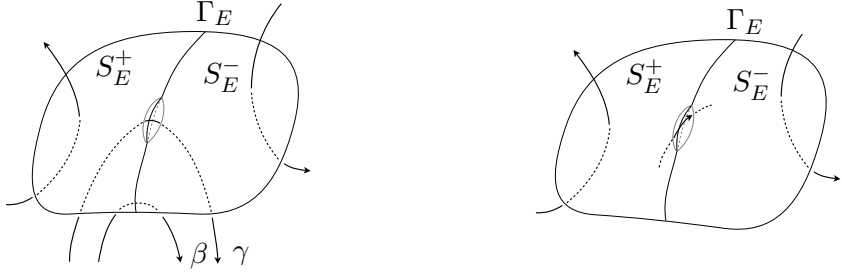

\begin{minipage}[c]{1\textwidth} \begin{center}
\begin{minipage}[c]{0.4\textwidth}
	\def\JPicScale{0.3}
	\input{figures/figure1a.pst}
\end{minipage}
\hspace{0.02\textwidth}
\begin{minipage}[c]{0.4\textwidth}
	\def\JPicScale{0.3}
	\input{figures/figure1b.pst}
\end{minipage}
\end{center} \end{minipage}
\caption{Proof of Theorem~\ref{thmLocMin}. Left: 
$\Phi \left(S_E \right)$ can be decreased if $\Gamma_E$ is not invariant under $X_H$, shown here by trajectories $\beta$ and $\gamma$ touching and crossing $S_E$, generated by a vector field $X_H$ not tangent to $\Gamma_E$, together with a deformation decreasing $\Phi \left(S_E \right)$. Right: $\Phi \left(S_E \right)$ can be decreased if a nearby orbit intersects $S_E$ twice in opposite directions.
}
\label{proofFig}
\end{figure}
Then, we are ready for 
\begin{thm}[\cite{MacKay1994}]
\label{thmLocMin}
A  codimension-1 orientable submanifold $S_E$ of an energy-level has locally minimal geometric flux if and only if it can be decomposed into surfaces $S^i_E$ of unidirectional stationary flux and $S_E$ has no local recrossings.
\end{thm}

\begin{proof}
Without loss of generality, we represent $S_E$ as the zero-set of some smooth function $G:M_E \rightarrow \mathbb{R}$ with $\rmd G \neq 0$ on $S_E$, so $S_E = \lbrace z \in M_E \vert G \left( z \right) = 0 \rbrace$.
To divide $S_E$ into unidirectional parts $S_E = \cup_i S_E^i$, we consider
$$\Gamma_E = \lbrace z \in M_E \vert \rmd G \left( X_H \right) = 0 \rbrace,$$
which gives $S_E \backslash \Gamma_E$ composed of parts $\hat{S}_E^i$. These are enlarged to $S_E^i$ including the invariant parts such that their union is the whole of $S_E$.

Now, assume that $S_E$ has locally minimal geometric flux, then it has stationary flux and its boundary $\del S_E$ is invariant, by Corollary \ref{statFlux}. The vector field $X_H$ is tangent to $\del S_E^i \backslash \del S_E$, otherwise we could deform $S_E$ and decrease $\Phi_E \left( S_E \right)$, see \figurename~\ref{proofFig}. Therefore $X_H$ is everywhere tangent to $\del S_E^i$, meaning that they are invariant and that $S_E^i$ have stationary flux, again by Corollary \ref{statFlux}. If $S_E^i$ has local recrossings, which must be near $\Gamma_E$, we can decrease $\Phi_E \left( S_E \right)$ by lifting $S_E$ locally, contradicting the assumption, see \figurename~\ref{proofFig}.

Conversely, assume that $S_E$ is the union of surfaces $S_E^i$ of unidirectional stationary flux. Then the flux through $S_E$ is the sum of those through $S_E^i$ and thus stationary, giving stationary geometric flux. The only places where a small deformation would make a difference to $\Phi_E \left( S_E \right)$ are near $\Gamma_E$. However, if there are no local recrossings, lifting $S_E$ near $\Gamma_E$ cannot decrease $\Phi_E \left( S_E \right)$. This can be shown by contradiction: if $\Phi_E \left( S_E \right)$ can be decreased by a small change, then there are points on $S_E$ whose orbit sneaks back to $S_E$.
\end{proof}


\subsection{Basic transport scenario: flux over a saddle}
\label{reactDyn}

This is the case of an autonomous Hamiltonian system $\left(M^{2m},\omega,H \right)$ with a non-degenerate index-1 critical point $\bar{z}_1$ of $H$.
It provides an example of a closed, invariant, codimension-2 submanifold of the energy levels spanned by two codimension-1 submanifolds of unidirectional flux with no local recrossings.

About $\bar{z}_1$, we have the Williamson normal form \cite{Williamson1936}
$$H \left(z \right) = E_{a} + \frac{\alpha}{2} \left(y^2 - x^2 \right) + \sum_{j=1}^{m-1} \frac{\beta_{j}}{2} \left(p_{j}^2 + q_{j}^2 \right) + H_n \left(z \right), \quad H_n \left(z \right) = \order{3},$$
where we ask that $\alpha$, $\beta_j > 0$, and $z = \left(x,q,y,p \right)$ are canonical coordinates with $ \left(x,y \right) $ the hyperbolic degree of freedom and $ \left(q,p \right) = \left(q_1,...,q_{m-1},p_1,...,p_{m-1} \right) $ the elliptic ones \cite[App.6]{arnold1989mathematical}. In the chemistry literature these are called the reaction and bath coordinates, respectively. Note that we do not need the higher-order normal forms found in some of the transition state theory literature, see e.g. Wiggins et al. \cite{wiggins2001impenetrable,uzer2002geometry}.

We now consider the topology of the energy levels about $\bar{z}_1$. First we consider them to second order, where they are given by
\begin{equation*}
M_{E} = H_2^{-1}\left( E \right) = \left\lbrace z \in M \bigg| \frac{\alpha}{2} \left( y^2 - x^2 \right) + \sum_{j=1}^{m-1} \frac{\beta_{j}}{2} \left(p_{j}^2 + q_{j}^2 \right) = \Delta E \right\rbrace ,
\end{equation*}
with $ \Delta E = E - E_{a} $. 
Therefore, in a neighbourhood of $\bar{z}_1$, we can write the energy level as the union of the graphs of two functions
$$x_\pm = \pm \sqrt{ \frac{2}{\alpha} \left( \frac{\alpha}{2} y^2 + \sum_{j=1}^{m-1} \frac{\beta_{j}}{2} \left(p_{j}^2 + q_{j}^2 \right) - \Delta E \right) },$$
over $\mathbb{R}^{2m-1}$. 
For $\Delta E < 0$,
$M_E$ is diffeomorphic to two copies of $\mathbb{R}^{2m-1}$, whereas for $\Delta E > 0$, the two disjoint regions connect and $M_E$ is diffeomorphic to $\mathbb{S}^{2m-2} \times \mathbb{R}$. This is a standard Morse surgery, see Theorem \ref{thmB} in Appendix \ref{morse}.
An important feature of the energy levels is the presence of a ``bottleneck'' about $z_1$, which opens up as the energy is increased from $\Delta E = 0$.
The two regions on either side of the critical point are the ones between which we want to study the transport of energy-surface volume.
For the topology of the energy-levels of the full system, we appeal to the Morse lemma, see Appendix \ref{morse}. This tells us that there are coordinates about the critical point, $z_1$, for which the Hamiltonian function is quadratic. Thus, the previous study of the quadratic case is sufficient. However, the transformation giving the Morse lemma coordinates is not necessarily symplectic. Therefore, whilst these coordinates can be used to study the topology of $M_E$, they cannot be used to study the dynamics without losing the simple expression of the Hamiltonian nature of the system.

\begin{rmk}
Note that considering a system with a saddle$\times$centre$\times \cdots \times$centre equilibrium, as is often stated, is not actually the same as considering an index-1 critical point of the Hamiltonian function for general Hamiltonian systems. We could have for example one unstable dimension and an arbitrary odd index, e.g. three with $\beta_1 < 0$ in the Williamson normal form, say.
Then $M_E$ does not separate for $\Delta E < 0$, i.e. the topology is different. However, for \textit{natural systems}\footnote{Systems for which the Hamiltonian function is the sum of kinetic and potential energy.} with positive-definite quadratic kinetic energy, these cases cannot arise, which is why the two situations are often confused.
\end{rmk}

We now find an invariant codimension-2 submanifold of the energy levels. 
The centre subspace $\hat{N} = \lbrace z \in M \vert x = y = 0 \rbrace$ of the linearised dynamics about $\bar{z}_1$ extends to a centre manifold $N$, which can locally be expressed in the form
$$\quad N = \lbrace z \in M \vert x = X \left(q,p \right), y = Y \left(q,p \right) \rbrace,$$
with the 1-jets of $X$ and $Y$ vanishing at $\bar{z}_1$.
Then $N_E = N \cap M_E$ is an invariant submanifold of the energy level.
$N_E$ is diffeomorphic to $\mathbb{S}^{2m-3}$. This is proved by 
using the Morse lemma as was done for the energy levels. The restriction of the Hamiltonian function to $N$ is
\begin{align*} 
H_N \left( q,p \right) &= E_{a} + \frac{\alpha}{2} \left(Y^2 - X^2 \right) + \sum_{j=1}^{m-1} \frac{\beta_{j}}{2} \left(p_{j}^2 + q_{j}^2 \right) + H_n \left( X,q,Y,p \right) \\
 &= E_{a} + \sum_{j=1}^{m-1} \frac{\beta_{j}}{2} \left(p_{j}^2 + q_{j}^2 \right) + \mathcal{O} \left(3 \right).
\end{align*}
Thus the origin, $\tilde{z}_1$, is a critical point of $H_N$ with index-$0$. 
Then by the Morse lemma, in a neighbourhood of $\tilde{z}_1$, we have $H_N \left( \tilde{z} \right) = E_a + \frac{1}{2} \left( y_1^2 + \cdots + y_{2m-2}^2 \right)$, so $ N_E = H_N^{-1} \left( E \right) \cong \mathbb{S}^{2m-3}$.
Finally, by Theorem \ref{thmA}, the diffeomorphism type is valid until the next critical value of $H_N$ (if one exists), and not just for small $\Delta E$, a proof of which can be found in Sacker \cite{Sacker}, and was already known to Conley and his students, see \cite{Easton1967, Conley1968, Sacker} and references therein.

We now have to show that $N_E$ can be spanned by surfaces $S_E^{\pm}$ of unidirectional flux. 
In a neighbourhood of $\bar{z}_1$, we can simply choose
$$S = \lbrace z \in M \vert G \left( z \right) = x - X \left(q,p \right) = 0 \rbrace ,$$
as done by Toller et al. \cite{Toller1985}, and intersect with $M_E$ to obtain $S_E$.
$S$ spans $N$, but we must check that it is the union of two surfaces of unidirectional flux. We decompose it into the parts $S^\pm$ with $y > Y \left(q,p \right)$ and $y < Y \left(q,p \right)$ and check that $ \rmd G \left( X_H \right) \ge 0 $ for $y > Y \left(q,p \right)$ and vice-versa. This will ensure that the halves of the dividing surface, $S_E^\pm$, are unidirectional, since the energy levels are invariant.
Firstly, we rewrite 
\begin{equation*}
\rmd G \left( X_H \right) = \lbrace G,H \rbrace = - \dot{G} \left(z \right),
\end{equation*}
then we find that
\begin{align*}
\dot{G} &= \dot{x} - DX \cdot \left(\dot{q}, \dot{p}\right) \notag\\
&= \alpha y + \partial_y H_n - DX \cdot \left(\dot{q}, \dot{p}\right).
\end{align*}
Now, the invariance of $N$, on which $x=X,$ $y=Y$, gives us that 
\begin{align*}
\alpha Y + \partial_y H_n^* - DX \cdot \left(\dot{q}, \dot{p} \right)^* &= 0,
\end{align*}
where the $^*$ denotes that the equation is evaluated on $N$.
Subtracting the (first) invariance equation from the flux equation gives
\begin{align*}
\dot{G} &= \alpha \left( y - Y \right) + \partial_y H_n - \partial_y H_n^* - DX \cdot \left( \left(\dot{q}, \dot{p}\right) - \left(\dot{q}, \dot{p}\right)^* \right) \notag\\
&= \alpha \left( y - Y \right) + \partial_y H_n - \partial_y H_n^* - DX \cdot \left( \partial_p H_n - \partial_p H_n^*, - \partial_q H_n + \partial_q H_n^* \right),
\end{align*}
where, for a given choice of $\left( q,p \right)$, the terms evaluated on $N$ are constant.
In a neighbourhood of the critical point $\bar{z}_1$, the first term dominates the others and gives the sign, since $H_n$ denotes the higher order terms in the Hamiltonian function, and $X$ defines the centre manifold and is second order in $\left( q, p \right)$.
This construction is local about $\bar{z}_1$. A neater construction, semi-local about $N$, will be presented in Section \ref{variation}.

\begin{rmk}
Note that the dividing surface $S_E$ constructed above is closed, and the two halves $S_E^\pm$ are compact surfaces with boundary $N_E$. On the other hand, the choice $S = \lbrace z \in M \vert y = Y \left( q, p \right) \rbrace$ would not have given compact intersections with $M_E$.
\end{rmk}

Now, in order to apply Theorem \ref{thmLocMin} and show that our dividing surfaces have locally minimal geometric flux, we require that the dividing surfaces have no local recrossings. However, the centre manifold $N$ is normally hyperbolic\footnote{Actually $N$ is not (necessarily) compact, but the level sets $N_E$ are invariant, so the sub-level sets $N_{\leq E}$ are compact submanifolds with (invariant) boundary $N_E$ and normally hyperbolic.} and it has stable and unstable manifolds $W^\pm$ of codimension-1 in $M$. Then $N_E$ is also normally hyperbolic, as a submanifold of $M_E$, and $W_E^\pm$ are codimension-1 in $M_E$, thus dividing a neighbourhood of $N_E$ into four sectors. Finally, since $S_E^\pm$ lie between $W_E^\pm$, unidirectionality implies that there are no local recrossings. 

\begin{nota}
There is a simple asymptotic law for the flux when $\Delta E$ is small \cite{MacKay1990}. In this case, $N_E$ is given to leading order by $x = q_m = 0, y = p_m = 0$, and $\sum_{j=1}^{m-1} \frac{\beta_{j}}{2} \left(p_{j}^2 + q_{j}^2 \right) = \Delta E$. The generalised action is 
\begin{align*}
\Lambda &= \frac{1}{\left(m-1\right)!} \lambda \wedge \omega^{m-2} \\
&= \frac{1}{\left(m-1\right)!} \left( p_1 \rmd q_1 \wedge \rmd q_2 \wedge \rmd p_2 \wedge \cdots \wedge \rmd q_{m-1} \wedge \rmd p_{m-1} + \text{ similar terms}\right).
\end{align*}
Thus, 
the flux is (cf. Vineyard \cite{Vineyard1957})
$$\phi_E \left( S_E^+ \right) = \frac{\Delta E^{m-1}}{\left( m-1\right)!} \prod_{j=1}^{m-1} \frac{2 \pi}{\beta_j}.$$
\end{nota}

%
\begin{figure}
\begin{minipage}[c]{1\textwidth} 
\begin{center}
\begin{minipage}[c]{0.4\textwidth}
	\begin{pspicture}(5,4.1)
		\rput(2.5,1.55){\includegraphics[width=0.95\textwidth]{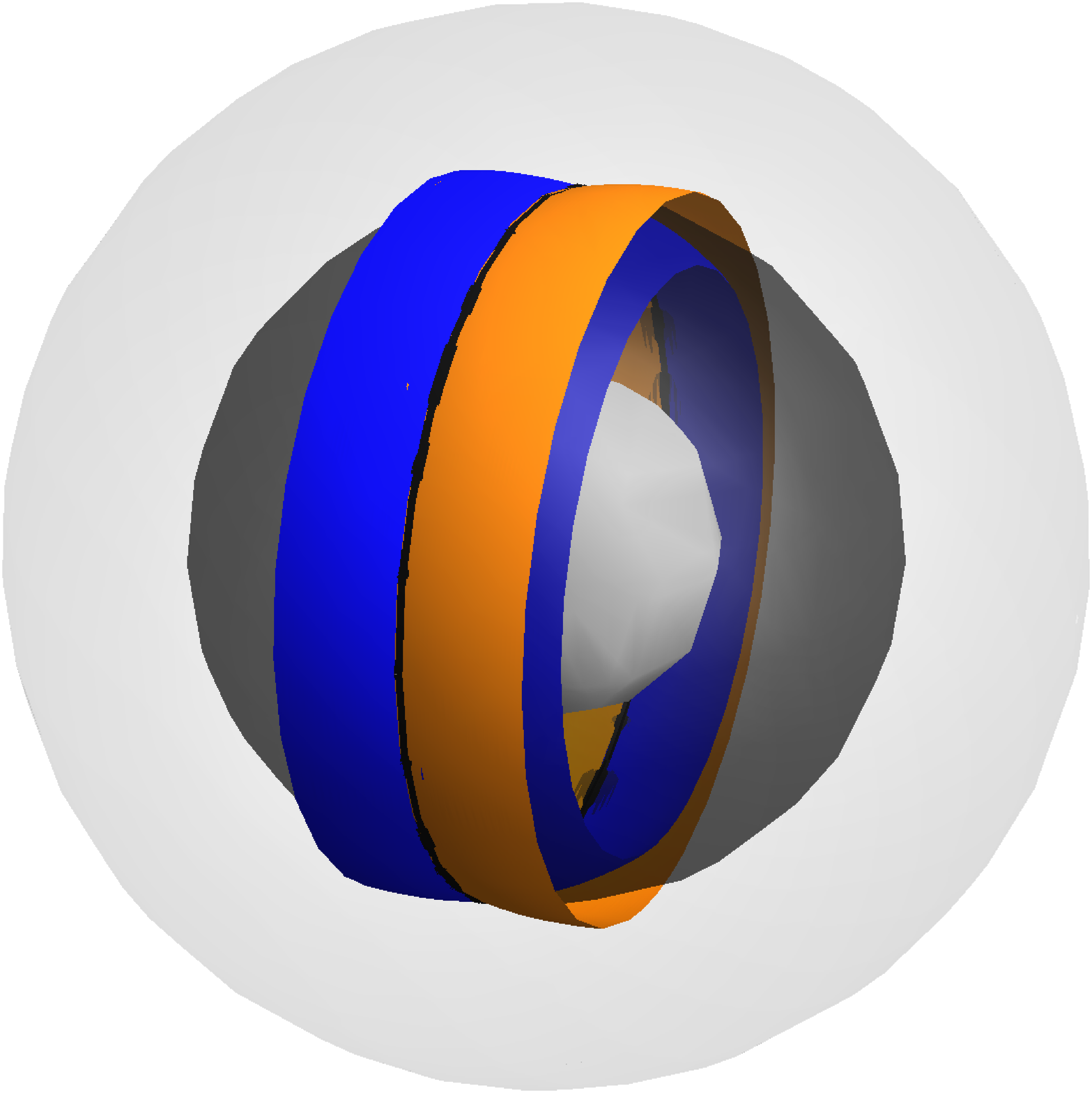}}
		\rput(0.9,3.5){$M_E$}
		\rput(2.2,1.7){\white $N_E$}
		\rput(1.12,1.9){$S_E^-$}
		\rput(4.39,1.9){$S_E^+$}
		\rput(2.5,3.76){$W_E^-$}
		\rput(3.45,3.63){$W_E^+$}
	\end{pspicture}
\end{minipage}
\hspace{0.02\textwidth}
\begin{minipage}[c]{0.4\textwidth}
	\begin{pspicture}(5,4.1)
		\rput(2.5,2){\includegraphics[width=0.9\textwidth]{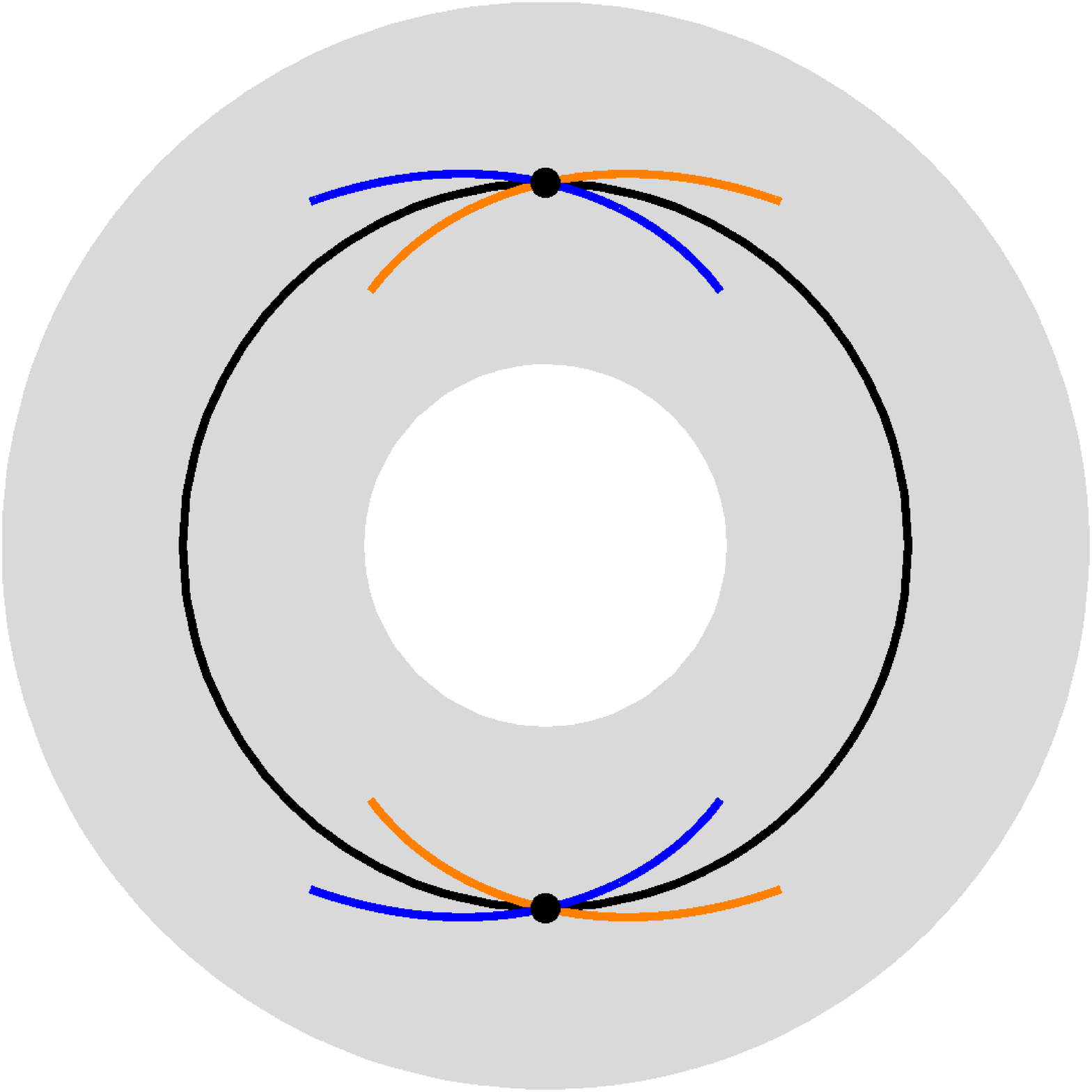}}
		\rput(0.9,3.5){$M_E$}
		\rput(2.5,0.40){$N_E$}
		\rput(1.30,1.1){$S_E^-$}
		\rput(3.82,1.1){$S_E^+$}
		\rput(1.67,3.65){$W_E^-$}
		\rput(3.2,3.65){$W_E^+$}
		\rput(3.45,2){$L_E$}
		\rput(4.2,2){$R_E$}
	\end{pspicture}
\end{minipage}
\end{center}
\end{minipage}
\caption{Conley representation of the quadratic approximation of the basic scenario, for some $E > E_a$, showing the energy-level $M_E$, the transition state $N_E$, its stable and unstable manifolds $W_E^\pm$ and the dividing surface $S_E = S_E^+ \cup S_E^-$. Left: full representation, right: cross-section with $u=0$.}
\label{conleyFig}
\end{figure}

A nice way of visualising the energy level and the various submanifolds is to use the Conley representation\footnote{The literature nowadays often also refers to it as the McGehee representation.
}. This method is implicit in a paper by Conley \cite{Conley1968}, was used by McGehee \cite{mcgehee1969homoclinic} and MacKay \cite{MacKay1990}, and is illustrated in \cite{Waalkens2010}.
Considering a 2 degree of freedom system and forgetting the higher order terms, the energy level $M_E$ is given by the equation
$$ \frac{\alpha}{2} y^2 + \frac{\beta}{2} \left(p^2 + q^2 \right) = \Delta E + \frac{\alpha}{2} x^2, \quad \Delta E = E - E_{a}.$$
We have seen that for $\Delta E > 0$, $M_E$ is diffeomorphic to $\mathbb{S}^2 \times \mathbb{R}$. The idea is therefore to represent $M_E$ as a spherical shell in $\mathbb{R}^3$ by considering it to be a 1-parameter ($x$) family of 2-spheres, which we denote $M_E^x$. For a given $x$, we project the sphere $M_E^x$ to another sphere in $\mathbb{R}^3$ by
$$ \pi_x : M_E^x \rightarrow \mathbb{R}^3 : \left( q,y,p \right) \mapsto \frac{r\left(x \right)}{r_E \left(x \right)} \left( q,y,p \right) =: \left( u,v,w \right),$$
where $ r_E \left(x \right) = \left(2 \Delta E + \alpha x^2 \right)^{1/2}$ and the new radius $ r\left(x \right)$ is a monotone function mapping the real line to a bounded positive interval, e.g. $ r\left(x \right) = 2 + \tanh \left( x\right) $, for which $r\left(x \right) \in \left[1,3 \right]$. Under this projection, the parameterised 2-spheres $M_E^x$ are placed concentrically in $\mathbb{R}^3$.
Then we define a map taking points on $M_E$ to $\mathbb{R}^3$ by
$$ \pi : \left( x,q,y,p \right) \mapsto \pi_x \left( q,y,p \right) = \left( u,v,w \right),$$
which gives the desired spherical shell.
The Conley representation of the quadratic approximation can be seen in \figurename~\ref{conleyFig}. This figure is only for 2 degree of freedom systems, but for $m$ degrees of freedom, the same procedure can be applied and we can imagine projecting $M_E$ to $\mathbb{R}^{2m-1}$.


\section{Transition states and dividing surfaces}
\label{transition}

We just saw how closed, invariant, codimension-2 submanifolds of the energy levels are the key to constructing surfaces with locally minimal geometric flux in the basic scenario.
As we are interested in what happens in the basic scenario when the energy is increased further, and in transport scenarios that are not governed by a local analysis, we give the general
\begin{defn}A \textit{transition state} for a Hamiltonian system is a closed, invariant, oriented, codi-mension-2 submanifold of an energy-level that can be spanned by two surfaces of unidirectional flux, whose union divides the energy-level into two components and has no local recrossings.
\end{defn}
``Transition state'' is not an ideal name because it is a set of states, not a single one. Furthermore, it is not a set of intermediate states on paths from reagents to products like the dividing surface, because it is invariant. The chemistry literature often also uses the term for dividing surfaces. This is why MacKay \cite{MacKay1990} avoided using the term.
However, it is by now established terminology and we choose to stick with tradition.

In the basic scenario, the transition states are level sets of the Hamiltonian function restricted to the centre manifold. 
The latter is normally hyperbolic. It is also symplectic, meaning that the restriction $\omega_N$ of the symplectic form to $N$ is non-degenerate. It is furthermore unique since centre manifolds are locally unique when their dynamics are bounded to a neighbourhood of the critical point for all time, see Sijbrand \cite[Thm3.2]{Sijbrand1985}. This is the case for the basic scenario centre manifold, where the motion takes place on invariant spheres $N_E \cong \mathbb{S}^{2m-3}$, the transition states, due to conservation of energy and positive definiteness of the restricted Hamiltonian function.
The normal hyperbolicity of the centre manifold as a submanifold of state space ensures that of the transition states within an energy level, provided that they are smooth manifolds. This in turn prevents local recrossings of the dividing surfaces.
Thus, for higher energies, we must consider normally hyperbolic extension of the centre manifold beyond a local neighbourhood of the index-1 critical point. We expect that these too will be symplectic\footnote{J-P Marco instead has results showing that normally hyperbolic submanifolds of Hamiltonian systems which also satisfy some extra conditions are symplectic.} because of experience with the axi-symmetric case of Section \ref{volcano}, for example, that is
\begin{conj}%
\label{claim}
If a connected submanifold $N$ of a Hamiltonian system is normally hyperbolic and $\omega_N$ is non-degenerate at one point of $N$ then it is non-degenerate on the whole of $N$.%
\end{conj}%

For more general transport scenarios, we expect invariant, codimension-2 submanifolds of state space, composed of the union of transition states over an interval of energy. These, and the centre manifolds of the basic scenario, we will refer to as \textit{transition manifolds}\footnote{In the chemistry literature, which until now has focused on the basic scenario, these are referred to as the ``activated'' complexes, states or surfaces, see e.g. Henriksen and Hansen \cite[p.140]{Henriksen2008}.}. 

We have been using \textit{dividing surface} to refer to codimension-1 submanifolds of an energy level that divide it into two parts. The union of the dividing surfaces with different energies will be referred to as a \textit{dividing manifold}. This is a codimension-1 submanifold, locally dividing state space.


\section{Constructing dividing manifolds}
\label{variation}

In the basic scenario, we have seen how to construct a local dividing manifold, about the critical point $\bar{z}_1$. More generally, given an invariant, codimension-2 submanifold of state space, we want a method to construct a dividing manifold spanning it. More precisely, let us say a potential transition manifold is a submanifold $N$ of state space $M$ that is invariant, orientable, codimension-2, and on which the restriction of the Hamiltonian $H_N$ is bounded from below and proper, ensuring compact sub-level sets, $N_{\leq E} = \lbrace z \in N \vert H_N \left( z \right) \leq E \rbrace$. The level sets of such submanifolds are closed, invariant, codimension-2 submanifolds $N_E$ of the energy levels $M_E$. Furthermore, we will ask that $N$ is symplectic, as in the basic scenario. This will be necessary in our construction, and seeing as the condition for a submanifold to be symplectic is open in the space of $C^1$ submanifolds of a given symplectic manifold, we do not feel that this imposes significant restrictions. Note that any symplectic submanifold is automatically orientable. Thus, provided $N$ is normally hyperbolic and its stable and unstable manifolds are orientable\footnote{This is not necessarily the case. An example of an orientable, codimension-2 normally hyperbolic submanifold with non-orientable stable and unstable manifolds is the orbit cylinder formed by a family of inversion hyperbolic periodic orbits (with negative characteristic multipliers) parametrised by the energy in a 2 degree of freedom system. This has local stable and unstable manifolds diffeomorphic to a M\"obius strip cross an interval, and emerges, for example, out of a period doubling bifurcation of an elliptic periodic orbit, see e.g. \cite[p.599]{abraham1987foundations}.}, we will show how to construct a codimension-1 submanifold $S$ locally dividing $M$ with no local recrossings and composed of two halves that span $N$ and across which the flow is unidirectional. Then, restricting $S$ to an energy level will give a dividing surface $S_E$ with locally minimal geometric flux, demonstrating that $S$ is a dividing manifold, $N$ a transition manifold and $N_E$ a transition state.

The construction requires a \textit{fibration} of a neighbourhood $U \subset M$ of $N$. This is a manifold $U$ (called the \textit{total space}) together with a projection $\pi: U \rightarrow N : z \mapsto \tilde{z}$ to a manifold $N$ (the \textit{base space}) such that the \textit{fibres} $F_{\tilde{z}} = \pi^{-1} \left( \tilde{z} \right)$ are submanifolds, and a local trivialisation $\psi_i : \pi^{-1} \left( V_i \right) \rightarrow V_i \times F$, where $V_i$ is a set in an open covering of $N$, and $F$ is a fixed manifold (the \textit{standard fibre}). It is usually denoted $\left( U , N, \pi, F \right)$.
The tangent spaces to the fibres give a \textit{vertical subbundle}, $\text{Vert}$, of the tangent bundle $TM$ that is tangent to the fibres
$$\text{Vert}_z = \ker \rmd_z \pi = T_z F_{\tilde{z}},$$
for all points $z \in U$. Then, a choice of \textit{horizontal subbundle}, $\text{Hor}$, gives a splitting of the tangent bundle
$$T M = \text{Vert} \oplus \text{Hor}.$$

Note that the fibres are 2-dimensional, the codimension of $N$ in $M$. We require the fibres to be symplectically orthogonal to the symplectic base space, and choose to consider a \textit{symplectic fibration}, for which the fibres are symplectic submanifolds of the total space, see e.g. Guillemin, Lerman and Sternberg \cite[Ch.1]{Guillemin1996}. With our symplectic total space, we can choose the symplectic form of the fibre $F_{\tilde{z}}$ to be $\omega_{F_{\tilde{z}}}$, the restriction of $\omega$. We then say that $\omega$ is \textit{fibre-compatible}. Asking that the fibration is symplectic adds a constraint, but in exchange we can associate to $\omega$ a symplectic splitting, by defining the horizontal subbundle to be symplectically orthogonal to the vertical subbundle, i.e.
$$ \text{Hor}_z = \text{Vert}_z^\omega := \lbrace \xi \in T_z M \vert \symp{\xi}{\eta} = 0 \; \forall \eta \in \text{ Vert}_z \rbrace. $$

Now, any point $\tilde{z} \in N$ is an index-1 critical point of the Hamiltonian function restricted to the symplectically orthogonal fibre at $\tilde{z}$, $H_{F_{\tilde{z}}}$. This is because $N$ is a codimension-2, normally hyperbolic submanifold and the splitting is symplectic, so the restricted vector field is Hamiltonian, that is $X_H \vert_{F_{\tilde{z}}} = X_{ H_{ F_{\tilde{z}} } }$ where $i_{X_{H_{F_{\tilde{z}}}}} \omega_{F_{\tilde{z}}} = \rmd H_{F_{\tilde{z}}}$, and  $\tilde{z} \in F_{\tilde{z}}$ is a hyperbolic equilibrium of the restricted flow. Thus, we can choose vector fields $X_\pm$ tangent to the fibres such that
$$\symp{X_-}{X_+}>0, \quad \lieL_{X_-} \lieL_{X_-} H < 0 \text{ and } \lieL_{X_+} \lieL_{X_+} H > 0.$$
Note that the Lie derivative $\lieL_X f$ of a $0$-form (i.e.~a function) $f$ is just $\lieL_X f = X \left( f \right) = \rmd f \left( X \right)$, but this last notation does not lend itself to being applied twice. 

We then define 
$$S_{\tilde{z}} = \lbrace z \in F_{\tilde{z}} \vert \lieL_{X_-} H \left( z \right) = 0 \rbrace \text{ and } S = \cup_{\tilde{z}} S_{\tilde{z}}.$$
This dividing manifold $S$ spans $N$ and is an orientable, codimension-1 submanifold of $M$. Orientability following from that of the stable and unstable manifolds. However, we must check the flux of state space volume across $S$. Note that if the state space flux is unidirectional, then so is the energy-surface flux, since the energy-levels are invariant. The transverse component of $X_H$ across $S_{\tilde{z}}$ is
$$\rmd \left( \lieL_{X_-} H \right) \left( X_H \right) = \lieL_{X_H} \lieL_{X_-} H.$$
To find its sign, firstly we note that  $\lieL_{\left[X_H, X_- \right]} H = \lieL_{X_H} \lieL_{X_-} H - \lieL_{X_-} \lieL_{X_H} H = \lieL_{X_H} \lieL_{X_-} H,$ since $\lieL_{X_H} H = \rmd H \left( X_H \right) =0$. Here $ \left[X,Y \right] = XY - YX$ is the Lie bracket of vector fields (thought of as differential operators, as in $X \left( f \right) = \rmd f \left( X \right)$). Next, $\lieL_{X_+} H$ is single signed across each half of $S_{\tilde{z}}$ because ${\tilde{z}} \in N$ is a critical point of $\lieL_{X_+} H$, but $\lieL_{X_+} \lieL_{X_+} H < 0$ on the whole of $F_{\tilde{z}}$, thus $\tilde{z}$ is a minimum. Therefore, we ask that
$$ \lieL_{\left[X_H, X_- \right]} H = - c_{\tilde{z}} \; \lieL_{X_+} H , \quad c_{\tilde{z}} \in \mathbb{R}_+,$$
which is compatible with the initial assumptions.

In practice, it is easier to check the conditions if the vector fields $X_\pm$ are Hamiltonian, so we choose functions $A^\pm_{\tilde{z}} : F_{\tilde{z}} \rightarrow \mathbb{R}$ and define $X_\pm = X_{ A^\pm_{\tilde{z}} } $, where $i_{X_{A^\pm_{\tilde{z}}}} \omega = \rmd A^\pm_{\tilde{z}}$. Then using Poisson brackets, defined as $ \lbrace A,B \rbrace = \symp{X_A}{X_B}$ for two functions on $M$, and considering $A^\pm_{\tilde{z}}$ as functions on the whole of $M$, we can rewrite the conditions satisfied by the vector fields as
$$\lbrace A^-_{\tilde{z}}, A^+_{\tilde{z}} \rbrace>0, \quad \lbrace \lbrace H, A^-_{\tilde{z}} \rbrace, A^-_{\tilde{z}} \rbrace < 0 \text{ and } \lbrace \lbrace H, A^+_{\tilde{z}} \rbrace, A^+_{\tilde{z}} \rbrace > 0,$$
and the new conditions, ensuring that the flux is unidirectional, as
$$ \lbrace A^-_{\tilde{z}}, A^+_{\tilde{z}} \rbrace > 0, \quad \lbrace H, A^-_{\tilde{z}} \rbrace = c_{\tilde{z}} \; A^+_{\tilde{z}} \text{ and } \lbrace \lbrace H, A^+_{\tilde{z}} \rbrace, A^+_{\tilde{z}} \rbrace >0.$$
Thus, we have actually found that
$$S_{\tilde{z}} = \lbrace z \in M \vert A^+_{\tilde{z}} \left(z \right) = 0 \rbrace.$$

Now, seeing as $\tilde{z} \in N$ is an index-1 critical point of the Hamiltonian function restricted to $F_{\tilde{z}}$, it has Williamson normal form
$$ H_{F_{\tilde{z}}} \left( x,y \right) = \frac{\alpha_{\tilde{z}}}{2} \left(y^2 -x^2 \right) + \order{3}, \quad \alpha_{\tilde{z}}  \in \mathbb{R}_+,$$
about $\tilde{z}$. We can then choose
$A^-_{\tilde{z}} \left( z \right) = y$, $ A^+_{\tilde{z}} \left( z \right) = - x$, for example.

The dividing surfaces $S_E$ are then simply given by intersecting the dividing manifold with the energy levels. We must check that these are closed and have no local recrossings.  In order to check that $S_E$ is closed, we show that the sub-level set $S_{\leq E} = \lbrace z \in S \vert H_S \left( z \right) \leq E \rbrace$ is compact. To check that $S_{\leq E}$ is compact, we restrict the fibration to $\pi \vert_{S_{\leq E}} : S_{\leq E} \rightarrow N_{\leq E}$, which has a compact base space $N_{\leq E}$ by choice, and compact fibres $S_{\leq E, \tilde{z}}$ as we can see from $H_{F_{\tilde{z}}}$ in normal form, and thus a compact total space $S_{\leq E}$, as desired. The dividing surface $S_E$ does not have local recrossings by the same argument as for the local dividing surfaces in the basic scenario. $N_E$ is a normally hyperbolic submanifold of $M_E$, and $W^\pm_E$ are codimension-1. Then for each $\tilde{z} \in N$, $W^\pm_E \left( \tilde{z} \right) = \lbrace z \in M \vert y \mp x = 0, \; H \left( z \right) = E \rbrace$, so $S^\pm_E$ lie in-between $W^\pm_E$ and unidirectionality implies no local recrossings. 

This construction generalises, and reduces to, Toller et al.'s local construction \cite{Toller1985} for the basic scenario
$$S = \lbrace z \in M \vert G \left( z \right) = x - X \left( q,p \right)  = 0 \rbrace,$$
see Section \ref{reactDyn}. In this case, the fibres symplectically orthogonal to $N$ are given by 
$$F_{\tilde{z}} = \Big\lbrace q - \tilde{q} = \tilde{X}_p \left[ y - \tilde{Y} \right] - \tilde{Y}_p \left[ x - \tilde{X} \right] , \;
p - \tilde{p} = \tilde{X}_q \left[ y - \tilde{Y} \right] - \tilde{Y}_q \left[ x - \tilde{X} \right] \Big\rbrace,$$
where $\tilde z = \left( \tilde{X},\tilde{q},\tilde{Y},\tilde{p} \right)$ is a point in $ N$, $\tilde{X} = X \left( \tilde{q}, \tilde{p} \right)$, $\tilde{X}_i = \del_i X \left( \tilde{q}, \tilde{p} \right)$ and similarly for $Y$. The rest follows.


\section{Bifurcations of transition states}
\label{bifurcations}

Many transport scenarios, including the basic one for energies significantly above the saddle, cannot be considered locally about a critical point. The picture is therefore more complicated than the simple one for flux over a saddle.

Starting with the easiest case, that of systems with 2 degrees of freedom, we know that the transition states, being closed and 1-dimensional, are periodic orbits. Thus, their possible bifurcations are well known, and can be found in the literature, e.g. Meyer et al. \cite[Ch.11]{meyer2009introduction} and Han\ss mann \cite[Ch.3]{hanssmann2007local}.
A crucial feature of the transition state in the basic scenario is its normal hyperbolicity, which ensures that dividing surfaces constructed about it have locally minimal geometric flux. This may be lost at higher energies. However, for periodic orbits, we know what to expect when normal hyperbolicity is lost because normally elliptic periodic orbits also persist, see e.g. Meyer et al. \cite[Ch.9]{meyer2009introduction} on the continuation of periodic orbits.
These bifurcations however affect the underlying transport problem, see e.g.~\cite{Pollak1978, Davis1987}.

For higher degrees of freedom, the transition states are normally hyperbolic submanifolds of higher dimensions, e.g. in the basic scenario we have a $\left(2m-3\right)$-sphere. The bifurcation of such submanifolds when they lose normal hyperbolicity is not much explored.
There have been studies proposing different approaches and partial normal form methods, see \cite{li2009bifurcation, Teramoto2011, Allahem2012} and references therein.
Nonetheless, bifurcations involving the loss of normal hyperbolicity are not well understood.

What has been overlooked though is that there is a large class of systems for which the transition state develops singularities, i.e. points at which the manifold structure fails, at some energy $E_b$ and then reforms as a non-diffeomorphic normally hyperbolic submanifold for energies above $E_b$. The dividing surfaces also undergo a similar bifurcation. In this case, we can say exactly what happens. The context for these bifurcations is that there is a normally hyperbolic submanifold in the full state space, the transition manifold, denoted $N$. 
For example, starting from the basic scenario the transition manifold is an extension of the centre manifold beyond a local neighbourhood of the index-1 critical point. The transition states are then the level-sets of the Hamiltonian function restricted to the transition manifold, $N_E = H_N^{-1} \left(E \right)$, and they undergo a Morse bifurcation. This occurs when $H_N$ has a critical point, and can be studied using Morse theory, see Appendix \ref{morse}.

The critical points of the restricted Hamiltonian function $H_N$ are also critical points of the Hamiltonian function. For the Morse bifurcations, these will be of index one or higher relative to $H_N$, and hence of index two or higher relative to $H$. There have been studies regarding the role of higher index (than one) critical points in transport problems, see e.g. Ezra and Wiggins \cite{ezra2009phase}, Collins, Ezra and Wiggins \cite{Collins2011a}, Haller et al. \cite{Haller2011}. These have however focused on the higher index critical points and a neighbourhood about these, and thus not searched for the global submanifolds beyond this neighbourhood.

The effect that the Morse bifurcations have on the flux of energy-surface volume through the dividing surface is considered in Appendix \ref{fluxCh}.


\subsection{Example. Disconnecting transition states}
\label{volcano}

We now turn to our first example of a Morse bifurcation. This shows one way in which the basic scenario transition state can change topology as the energy increases.

\begin{figure}
	\centering
	\begin{pspicture}(8,4.5)
		\rput(4,2){\includegraphics[width=0.5\textwidth]{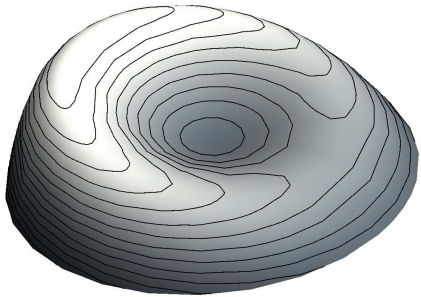}}
	\end{pspicture}
	\centering
	\caption{Graph of volcano potential with contour lines for the disconnecting transition states example.}
	\label{VPotFig}
\end{figure}

Consider a ``volcano potential'' given in polar coordinates as
\begin{equation*}
U\left( r, \theta \right) = \frac{1}{2} r^2 \left( 2 - r^2 \right) \left( 1 - \varepsilon r \cos \theta \right),
\end{equation*}
where $\varepsilon$ is a small positive parameter, see \figurename~\ref{VPotFig}.

The natural Hamiltonian function with this potential energy is then
\begin{equation*}
H\left( r, \theta, p_r, p_\theta \right) =  \frac{1}{2} \left(p_r^2 + \frac{1}{r^2}p_\theta^2 \right) + \frac{1}{2} r^2 \left( 2 - r^2 \right) \left( 1 - \varepsilon r \cos \theta \right),
\end{equation*}
and we consider the Hamiltonian system with the canonical symplectic form. The Hamiltonian function has three critical points, one of which $\bar{z}_1 = \left( r_\varepsilon, 0,0,0 \right)$ with index-1 , and another $\bar{z}_2 = \left( r_\varepsilon, \pi, 0,0 \right)$ with index-2. We are interested in transport in and out of the crater and
have an index-1 critical point in between, as expected for the basic scenario. Choosing $r$ as the coordinate joining the two regions, we want to construct a transition and dividing manifold about the index-1 critical point and study the transition states and dividing surfaces over a range of energies.

A similar volcano potential is seen in the ionization of hydrogen in a circularly polarized microwave field studied by Farrelly and Uzer \cite{Farrelly1995}. The transport problem in their example is however different, as escaping from the volcano's crater does not necessarily imply ionization.

Considering the axi-symmetric case, $\varepsilon = 0$, for which a transition manifold can be found explicitly, allows us to find an approximate transition manifold for the full system and due to its normal hyperbolicity deduce that there is a true transition manifold nearby.
The set of critical points of the Hamiltonian function restricted to the fibres, symplectically orthogonal to the axi-symmetric transition manifold, gives an approximate transition manifold as explained in Appendix \ref{nhms}.
Then the Morse bifurcations of the approximate transition states will be qualitatively the same as those of the actual transition states.

In the axi-symmetric case, the Hamiltonian function becomes
\begin{equation*}
H_0 \left( z \right) =  \frac{1}{2} \left(p_r^2 + \frac{1}{r^2} p_\theta^2 \right) + \frac{1}{2} r^2 \left( 2 - r^2 \right),
\end{equation*}
where $\theta$ is a cyclic coordinate, so the angular momentum is conserved, $p_\theta= \mu$. The critical points are $\bar{z}_0 = \left(0,\theta,0,0 \right)$ and $\bar{z}_1 = \left(1,\theta,0,0 \right)$, which are now both degenerate. 
Linearising about $\bar{z}_{1}$, we find that $\left(r, p_r \right)$ are the hyperbolic directions, and $\left(\theta, p_\theta \right)$ the elliptic ones. Thus, the centre subspace about $\bar{z}_1$ is $ \hat{N} = \lbrace z \in M \vert r = 1, \; p_r=0 \rbrace $, and we find the centre manifold, $ N_0 = \lbrace z \in M \vert r = \rho_0 \left( \theta, p_\theta \right), \; p_r=P_0 \left( \theta, p_\theta \right) \rbrace $ by satisfying the invariance equations
\begin{align*}
P_0 - \frac{p_\theta}{\rho_0^2}  \frac{\del \rho_0}{\del \theta} &=0, \\
\frac{p_\theta^2}{\rho_0^3}  - 2 \rho_0 \left( 1 - \rho_0^2 \right) - \frac{p_\theta}{\rho_0^2} \frac{\del P_0}{\del \theta} &= 0.
\end{align*}
This is done by choosing $P_0=0$ and $\rho_0$ satisfying $p_\theta^{2} - 2 \rho_0^4 \left( 1- \rho_0^2 \right) = 0$.

We have actually found a generalised centre manifold that extends beyond a small neighbourhood of $\bar{z}_1$.
To check the stability of $N_0$, i.e.~that it remains normally hyperbolic, we need to find appropriate tangent and normal coordinates and consider the linearised equations about $N_0$. At a point $\tilde{z} = ( \tilde{\rho}_0, \tilde{\theta}, 0, \tilde{p}_\theta )$ on the transition manifold, the tangent vectors are taken to be
$\xi_1 = \del_\theta$, $\xi_2 = \tilde{p}_\theta \del_r + 2 \tilde{\rho}_0^3 \left(2 -3 \tilde{\rho}_0^2 \right) \del_{ p_\theta}$.
We then choose a Riemannian structure, for which the length is given by
$$\rmd s^2 = \frac{c^2}{r^2} \left( \rmd r^2 + r^2 \rmd \theta^2 \right) + \rmd p_r^2 + \frac{1}{r^2} \rmd p_\theta^2,$$
i.e. proportional to the length in configuration space plus the kinetic energy, 
where the constant $c$ balances the dimensions by having those of velocity, and is set to 1. This allows us to define vectors orthogonal to the transition manifold as
$\eta_1 = 2 \tilde{\rho}_0^3 \left(3 \tilde{\rho}_0^2 -2 \right) \del_r +  \tilde{p}_\theta \del_{p_\theta}$ and $\eta_2 = \del_{p_r}$.
Finally, the first variation equations for $ \nu = v_1 \xi_1 + v_2 \xi_2 + v_3 \eta_1 + v_4 \eta_2$ are
$$\dot{v} = 
\left(
\begin{array}{cccc}
 0 & \frac{1}{\tilde{\rho_0} d\left( \tilde{\rho_0} \right)} & 0 & \frac{1 + 8 \tilde{\rho_0}^2 - 12 \tilde{\rho_0}^4}{2 \tilde{\rho_0}^6 d\left( \tilde{\rho_0} \right)}\\
 0 & 0 & \tilde{p}_\theta & 0 \\
 0 & 0 & 0 & \frac{2}{\tilde{\rho}_0^3} \\
 0 & 0 & 2 \tilde{\rho}_0^3 (3 \tilde{\rho}_0^2 -2) & 0
\end{array}
\right)
v,$$
where $d\left( \tilde{\rho_0} \right) = 1 + 7 \tilde{\rho}_0^2 - 24 \tilde{\rho}_0^4 + 18 \tilde{\rho}_0^6$ and is negative in the region of interest.
This choice of splitting is not invariant, but we can see that $N_0$ is normally hyperbolic for $\rho_0 \in \left( \sqrt{ 2/3},1 \right]$.
The point $\rho_0=\sqrt{ 2/3}$ is the steepest point of the potential, at which the normally hyperbolic periodic orbit emanating from the critical point $\bar{z}_1$, that is the transition state, collides with the elliptic periodic orbit from the crater of the volcano in a centre-saddle bifurcation \cite[\S 3.1]{hanssmann2007local}. Interestingly, at $\rho_0=\sqrt{ 2/3}$ when normal hyperbolicity is lost, the symplectic form restricted to the transition manifold $\omega_N$ also becomes degenerate.

\begin{figure}
	\centering
	\begin{pspicture}(14,4.5)
		\psfrag{r}[tr][bl]{\footnotesize $r$}
		\psfrag{e}{\footnotesize $E$}
		\psfrag{x}[tr][bl]{\footnotesize $r$}
		\psfrag{y}[r][l]{\footnotesize $p_r$}
		\psfrag{a}[r][l]{\tiny $E_a$}
		\psfrag{b}[r][l]{\tiny $E_c$}
		\psfrag{0}[t][b]{\tiny $0$}
		\psfrag{1}[t][b]{\tiny $1$}
		\psfrag{z}[r][l]{\tiny $0$}
		\psfrag{c}[t][b]{\tiny $\sqrt{\frac{2}{3}}$}
		\psfrag{d}[t][b]{\tiny $1$}
		\psfrag{f}[r][l]{\tiny $\frac{1}{2}$}
		\psfrag{h}[r][l]{\tiny $-\frac{1}{2}$}
		\rput(2,2.2){\includegraphics[width=0.3\textwidth]{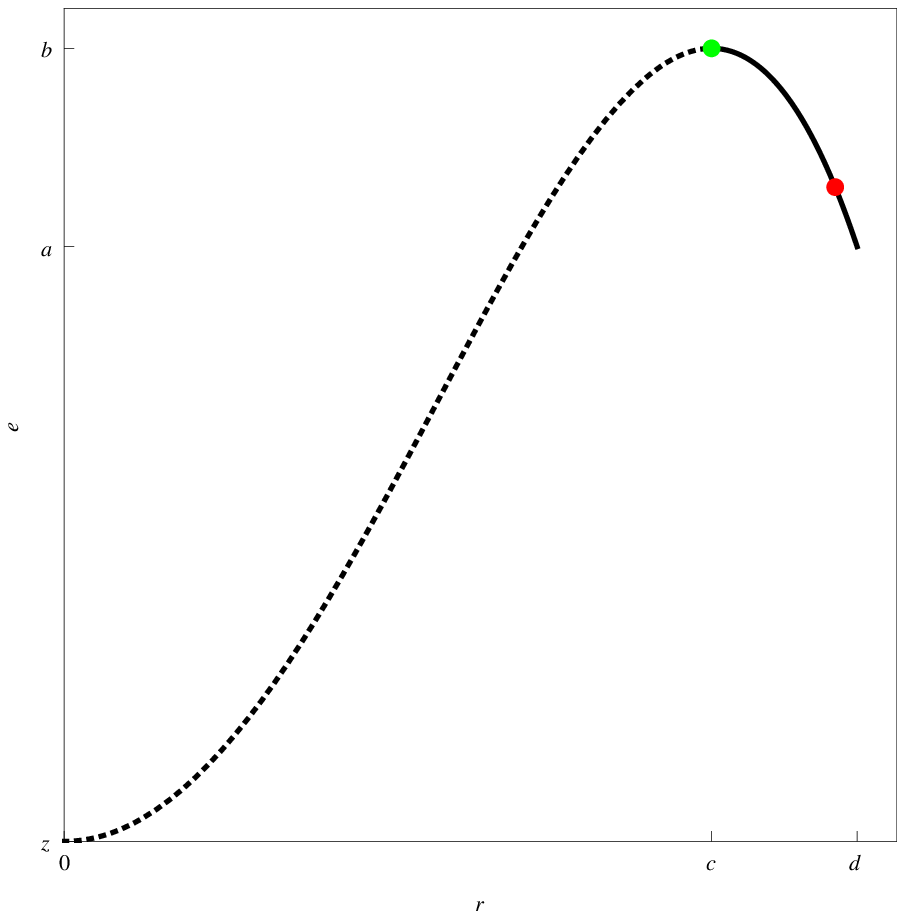}}
		\rput(7,2.2){\includegraphics[width=0.3\textwidth]{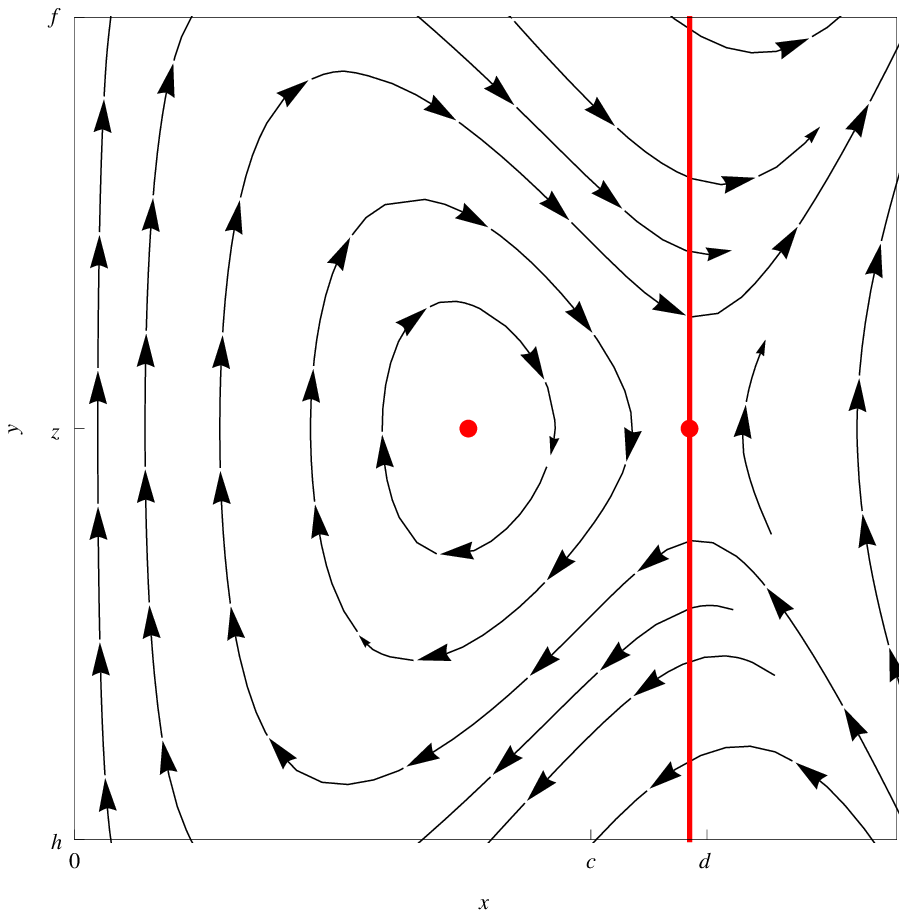}}
		\rput(12,2.2){\includegraphics[width=0.3\textwidth]{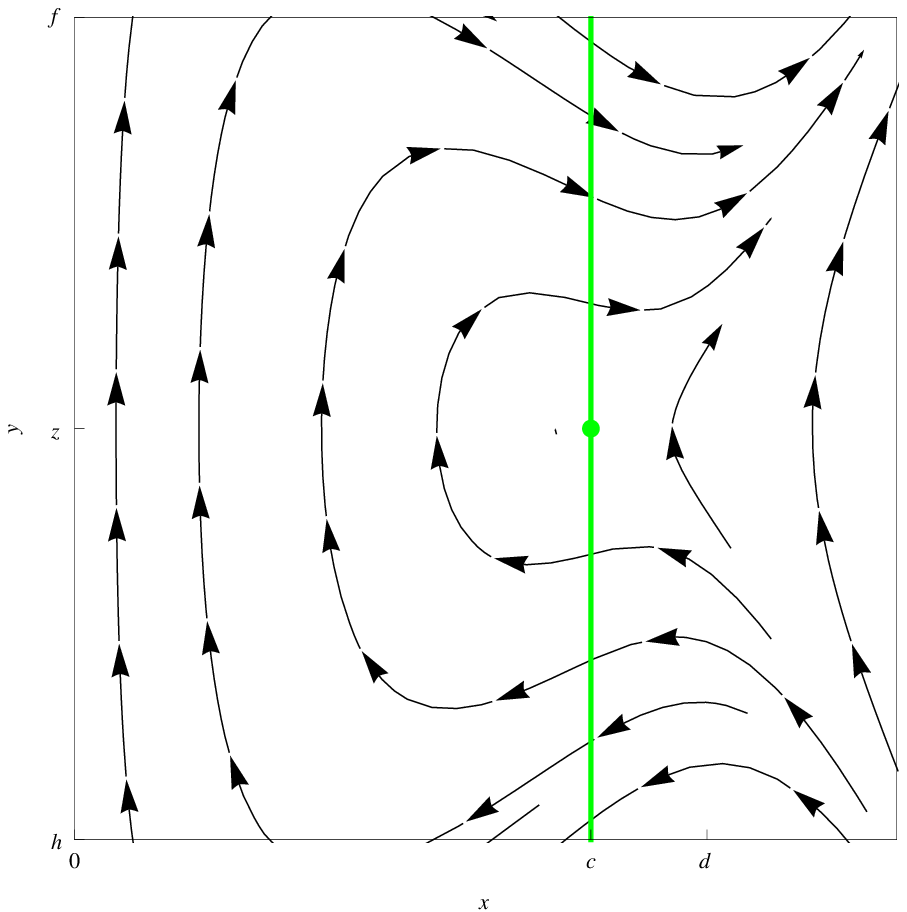}}
	\end{pspicture}
	\caption{Axi-symmetric case: energy of the transition states as a function of $r$ continued (dashed) to the elliptic periodic orbit, and flow in the $\left( r , p_r \right)$ plane for an energy in $(E_a, E_c)$ showing the transition state, the dividing surface and the flux through it, and for the bifurcational energy $E_c$.}
	\label{VDSFig}
\end{figure}

In the axi-symmetric case, we can use the dividing manifold construction method of Section \ref{variation}. The fibres, symplectically orthogonal to $N_0$ are
$$F^0_{\tilde{z}} = \lbrace z \in U \vert \theta = \tilde{\theta} + \frac{2 \tilde{\rho}_0^3 \left( 2-3 \tilde{\rho}_0^2 \right)}{\tilde{p}_\theta} p_r, \; p_\theta = \tilde{p}_\theta \rbrace.$$
Thus the restricted Hamiltonian function, linearised about $\tilde{z} \in N_0$ is 
$$ H_0 \vert_{F^0_{\tilde{z}}} \left( z \right) = E_0 + \frac{1}{2} y^2 - 2 \left( 3 \tilde{\rho}_0^2 - 2 \right) x^2 + \order{x^3}
.$$
The functions $A^-_{\tilde{z}} \left( z \right) = y$, $A^+_{\tilde{z}} \left( z \right) = -x$ satisfy the necessary conditions of Section \ref{variation}, and 
$$ S_{\tilde{z}} = \lbrace z \in U \vert A^+_{\tilde{z}} \left( z \right) = \tilde{\rho}_0 - r = 0 \rbrace.$$
Flow in the $\left( r, p_r \right)$ plane showing the dividing surfaces and the flux through them, for an energy in $(E_a, E_c)$ and for the bifurcational energy $E_c$, are shown in \figurename~\ref{VDSFig}. There are clearly recrossings even for small energies above $\bar{z}_1$ due to the geometry of the system, but these are not local (see definition in Section \ref{sectMin}). The local recrossings only appear when the transition state loses normal hyperbolicity at $\rho_0 = \sqrt{2/3}$.

Returning to the full system, we can now find an approximate transition manifold, $N_1$, by constructing a fibration of a local neighbourhood of $N_0$, with symplectically orthogonal fibres $F^0_{\tilde{z}}$, for $\tilde{z} = \left( \tilde{\theta}, \tilde{p}_\theta \right) \in N_0$, and then defining $ N_1 = \lbrace z \in M \vert \rmd_z H_{F^0_{\tilde{z}}}=0 \rbrace $. 
The symplectically orthogonal fibres are the ones used previously to find a dividing manifold for the axi-symmetric case. The Hamiltonian function restricted to the fibre $F^0_{\tilde{z}}$ is then
$$ H_{F^0_{\tilde{z}}} \left( r, p_r \right) = \frac{1}{2} \left( p_r^2 + \frac{1}{r^2} \tilde{p}_\theta^2 \right) + \frac{1}{2} r^2 \left( 2 - r^2 \right) \left( 1 - \varepsilon r \cos \left( \tilde{\theta} + \frac{2 \tilde{\rho}_0^3 \left( 2-3 \tilde{\rho}_0^2 \right)}{\tilde{p}_\theta} p_r \right) \right),$$
so linearising about $N_0$ (by letting $r = \tilde \rho_0 + \varepsilon \rho_1$, $p_r = \varepsilon P_1$) and taking the exterior derivative gives
\begin{align*}
\rmd_z H_{F^0_{\tilde{z}}} &= 
\frac{\varepsilon^2}{2}\left[ 8 \left( 2-3 \tilde \rho_0^2 \right) \rho_1 - \tilde \rho_0^2 \left( 6-5 \tilde \rho_0^2 \right) \cos \tilde \theta \right] \rmd \rho_1 \\
&+ \varepsilon^2 \left[ P_1 - \frac{\tilde \rho_0^6 \left( 4 - 8 \tilde \rho_0^2 + 3 \tilde \rho_0^4 \right)}{\tilde p_\theta} \sin \tilde \theta \right] \rmd P_1 + \order{\varepsilon^3}.
\end{align*}
Asking that $\rmd_z H_{F^0_{\tilde{z}}} = 0$, we obtain
\begin{align*}
\rho_1 &= \frac{\tilde \rho_0^2 \left( 6-5 \tilde \rho_0^2 \right)}{ 8 \left( 2-3 \tilde \rho_0^2 \right)} \cos \tilde \theta + \order{\varepsilon} \\
P_1 &= \frac{\tilde \rho_0^6 \left( 4 - 8 \tilde \rho_0^2 + 3 \tilde \rho_0^4 \right)}{\tilde p_\theta} \sin \tilde \theta + \order{\varepsilon}.
\end{align*}
and so
$$ N_1 = \lbrace z \in M \vert r = \rho_0 \left( p_\theta \right) + \varepsilon \rho_1 \left( \theta, p_\theta \right) + \order{\varepsilon^2} , \; p_r = 0 + \varepsilon P_1 \left( \theta, p_\theta \right) + \order{\varepsilon^2} \rbrace.$$
Then the restricted Hamiltonian is 
\begin{align*}
H_N \left( \theta, p_\theta \right) 
&= \frac{1}{2} \left( \varepsilon^2 P_1^2 + \frac{1}{\rho^2} p_\theta^2 \right) + \frac{1}{2} \rho^2 \left( 2 - \rho^2 \right) \left( 1 - \varepsilon \rho \cos \theta \right) \\
&= \frac{1}{2} \left( \frac{1}{\rho_0^2} p_\theta^2 + \rho_0^2 \left( 2 - \rho_0^2 \right) \left( 1- \varepsilon \rho_0 \cos \theta \right) \right) - \varepsilon \rho_1 \left( \frac{p_\theta^2}{\rho_0^3} - 2 \rho_0 \left(1 - \rho_0^2 \right) \right)  + \mathcal{O} \left( \varepsilon^2 \right) \\
&= \frac{1}{2} \frac{1}{\rho_0^2} p_\theta^2 + \frac{1}{2} \rho_0^2 \left( 2 - \rho_0^2 \right) \left(1- \varepsilon \rho_0 \cos \theta \right) + \mathcal{O} \left( \varepsilon^2 \right),
\end{align*}
and is actually independent of $\rho_1$ and $P_1$ to first order in $\varepsilon$. Note that we have dropped the subscript 1.
Finally, the transition states are given to order $\varepsilon$ as the level sets of the restricted Hamiltonian function, $N_E = H_N^{-1} \left(E \right)$.

The approximate dividing surfaces are then the level sets of an approximate dividing manifold chosen to be
$$S = \lbrace z \in M \vert r = \rho_0 \left( p_\theta \right) + \varepsilon \rho_1 \left( \theta, p_\theta \right) + \order{\varepsilon^2} \rbrace.$$
This spans the approximate transition manifold, which is not invariant, so it does not have minimal geometric flux and the two halves will not be unidirectional, in general. However, there are a true transition manifold and dividing manifold nearby. The true transition manifold being derived by normal hyperbolicity and the true dividing manifold by our construction of Section \ref{variation}.

\begin{figure}
	\centering
	\begin{pspicture}(7,5.5)
		\rput(3.5,2.6){\includegraphics[width=0.4\textwidth,height=0.4\textwidth]{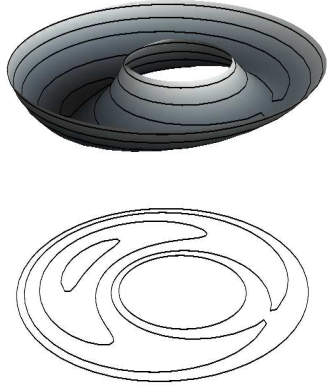}}
	\end{pspicture}
	\caption{Graph of the Hamiltonian function restricted to the transition manifold $H_N$, over an annulus in $\left(\theta , p_\theta \right)$, for the disconnecting example and its projections showing the transition states.
	}
	\label{VTMFig}
\end{figure}

We now consider the topology of the transition states and the dividing surfaces.
Starting with the transition state, we find that within the normally hyperbolic region, the restricted Hamiltonian function $H_N$ has critical points $\tilde{z}_1 = \left(0,0 \right)$ and $\tilde{z}_2 = \left(\pi,0 \right)$ with $\rho_0 = 1$. These have index $\tilde{\lambda}_1=0$ and $\tilde{\lambda}_2=1$, and energies $\frac{1}{2} \left(1-\varepsilon \right)$ and $\frac{1}{2} \left(1+\varepsilon \right)$, respectively.
Starting from the critical point with least energy, $\tilde{z}_1$, by the Morse lemma and Theorem \ref{thmA}, for energies below that at $\tilde{z}_2$ the transition state is diffeomorphic to a circle, $\mathbb{S}^1$. Increasing the energy and passing the critical point $\tilde{z}_2$ results in a bifurcation and the topology of $N_E$ changes, according to Theorem \ref{thmB}, to $2\mathbb{S}^1$, see \figurename~\ref{VTMFig}.
Thus, we have found our first example of a transition state bifurcation, and therefore of a transition state not diffeomorphic to $\mathbb{S}^{2m-3}$, namely $2\mathbb{S}^1 \ncong \mathbb{S}^1$. Similarly, we see that the dividing surface bifurcates and changes from a sphere $\mathbb{S}^2$ to a torus $\mathbb{T}^2$. It can be useful for extrapolation to higher degrees of freedom to write the transition state as $\mathbb{S}^0 \times \mathbb{S}^1$ ($\mathbb{S}^0$ being the two-point set $\lbrace \pm 1\rbrace$) and the dividing surface as $\mathbb{S}^1\times \mathbb{S}^1$.

Care must be taken in studying the Morse bifurcations, as the critical points of the restricted Hamiltonian functions are also critical points of the original Hamiltonian and therefore cause a change in the topology of the energy levels. In this example, the bottleneck opens up and the energy levels change topology, but we can still distinguish two regions and consider transport between them.

Morse theory applies to manifolds of all dimensions. This example can therefore be coupled to another (or more) oscillating degree of freedom to give a 3 degree of freedom system with Hamiltonian function
\begin{equation*}
H\left( r, \theta, q, p_r, p_\theta, p \right) =  \frac{1}{2} \left(p_r^2 + \frac{1}{r^2}p_\theta^2 \right) + \frac{1}{2} r^2 \left( 2 - r^2 \right) \left( 1 - \varepsilon r \cos \theta \right) + \frac{\beta}{2} \left( p^2 + q^2 \right) + \delta V \left( r, \theta, q \right).
\end{equation*}
In the uncoupled case, with $\delta=0$, there is no energy transfer with the new degree of freedom, so we can effectively consider the original volcano system and the oscillator separately. For energy above the maximum on the volcano rim in the volcano degree of freedom, the transition state bifurcates from $\mathbb{S}^3$ to $\mathbb{S}^2 \times \mathbb{S}^1$ and the dividing surface from $\mathbb{S}^4$ to $\mathbb{S}^3 \times \mathbb{S}^1$.
A small perturbation, $\delta \neq 0$, couples the degrees of freedom, but the normally hyperbolic transition manifold persists, along with the Morse bifurcation, so the same scenario occurs.
Specific  examples of higher degree of freedom systems exhibiting this Morse bifurcation will be seen in Section~\ref{bimolecular} where we consider bimolecular reactions.


\subsection{Example. Connecting transition states}
\label{ezra}

\begin{figure}
	\centering
	\begin{pspicture}(8,4)
		\rput(4,2){\includegraphics[width=0.45\textwidth]{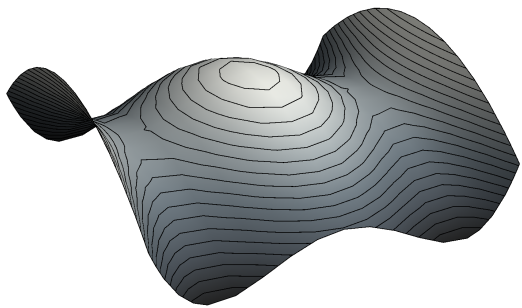}}
	\end{pspicture}
	\caption{Graph of the potential energy for the connecting transition states example.}
	\label{EPotFig}
\end{figure}

This example is found in applications such as narcissistic isomerisation reactions, that is chemical reactions in which a given molecule changes from one of its stereoisomers to the mirror image. References to this and other chemical reactions in which this bifurcation appears can be found in Ezra and Wiggins \cite{ezra2009phase}, where this example is also considered. They however focus on a neighbourhood of the index-2 critical point of the Hamiltonian function and the influence of this critical point on the transport, whereas we consider the complete picture.

The Hamiltonian system in question has $T^* \mathbb{R}^2$ as its state space, with its  canonical symplectic form and the Hamiltonian function
\begin{equation*}
H \left(z \right) = \frac{\alpha_2}{4} + \frac{\alpha_1}{2} \left(y^2 - x^2 \right) + \frac{\alpha_2}{2} \left(v^2 - u^2 \right) + \frac{\alpha_2}{4} u^4,
\end{equation*}
where $z = (x,u,y,v)$ and $\alpha_1, \alpha_2 \in \mathbb{R}^+$. The critical points of the Hamiltonian function are the origin, $\bar{z}_0$, and $\bar{z}_{\pm} = \left(0,\pm 1,0,0 \right)$, with index 2 and 1, respectively.

We are interested in transport between the two regions on either side of the two index-1 critical points and therefore the $x$-axis, see \figurename~\ref{EPotFig}. We therefore expand the Hamiltonian function about these critical points by shifting the $u$-axis, namely $u = \pm 1 + q$ and letting $v=p$, to get
\begin{equation*}
H \left(z \right) = \frac{\alpha_1}{2} \left(y^2 - x^2 \right) + \frac{\alpha_2}{2} \left(p^2 + 2 q^2 \right) + H_n \left( z \right) ,
\end{equation*}
with the higher order terms $ H_n \left( z \right) = \pm \alpha_2 q^3 + \frac{\alpha_2}{4} q^4$. Thus the centre subspaces of the critical points are seen to be $ \hat{N} \left( \bar{z}_\pm \right) = \lbrace z \in M \vert x=y=0 \rbrace $. Seeing as the system is uncoupled, the (local) centre manifolds can be chosen to be equal to the centre subspaces.

The two centre manifolds form part of a larger codimension-2 invariant submanifold, given by
\begin{equation*}
N = \lbrace z \in M \vert x=y=0 \rbrace,
\end{equation*}
for which we must check the stability, ensuring that we have normal hyperbolicity and so a transition manifold. This is done by linearising the vector field about $N$ and comparing the linear flows in the the normal, with $\eta_1 = 2^{-1/2} \left( \del_x - \del_y \right)$ and $\eta_2 = 2^{-1/2} \left( \del_x + \del_y \right)$, and the tangent, with $\xi_1 = \del_q$ and $\xi_2 = \del_p$, directions. The linearised equations of motion, about a point $\tilde{z}=\left(\tilde{q},\tilde{p} \right)$ in $N$, are
$$\dot{v} = 
\left(
\begin{array}{cccc}
 0 & \alpha_2 & 0 & 0 \\
 \alpha_2 - 3 \tilde{q}^2 & 0 & 0 & 0 \\
 0 & 0 & \alpha_1 & 0 \\
 0 & 0  & 0 & - \alpha_1
\end{array}
\right)
v, $$
where $ \nu= v_1 \xi_1 + v_2 \xi_2 + v_3 \eta_1 + v_4 \eta_2$. 
The normal dynamics are clearly hyperbolic. Instead, the dynamics tangent to $N$ depend on the point $\tilde{z}$ on the manifold. For $\left(\alpha_2 - 3 \tilde{q}^2 \right)<0$ the motion is elliptic, whereas for $\left(\alpha_2 - 3 \tilde{q}^2 \right)>0$ it is hyperbolic. Although in this uncoupled system we do not need $N$ to be normally hyperbolic, we do to continue the conclusions to cases with small coupling. We therefore compute a condition ensuring that the normal dynamics still dominates the tangent one. The coefficient $\left(\alpha_2 - 3 \tilde{q}^2 \right)$ is greatest when $\tilde{q} =0$, thus with $\alpha_1 > \alpha_2$ the transition manifold is normally hyperbolic. In the basic scenario, (half of) the normally hyperbolic degree of freedom gives the transport direction. At the critical point $\bar{z}_0$ however, the two directions ``compete'' because the tangent dynamics becomes hyperbolic. If the transition manifold stays normally hyperbolic, i.e. the potential energy (surface) is steepest in the $x$ direction, then the transport coordinate is preserved.

Finding a dividing manifold for this example is easy due to the lack of coupling. A simple fibration of a neighbourhood of $N$ has fibres given by $F_{\tilde{z}} = \lbrace u = v = 0 \rbrace$. Restricting the Hamiltonian function to such fibres gives the necessary normal form Hamiltonian. The functions $A^-_{\tilde{z}} \left( z \right) = y$, $A^+_{\tilde{z}} \left( z \right) = \lbrace H, A^-_{\tilde{z}} \rbrace = - \alpha_1 x$ then satisfy the necessary conditions of Section \ref{variation}. Therefore, a dividing manifold is given by
$$S = \lbrace z \in M \vert x = 0 \rbrace.$$

\begin{figure}
	\centering
	\begin{pspicture}(8,4)
		\rput(4,1.6){\includegraphics[width=0.5\textwidth]{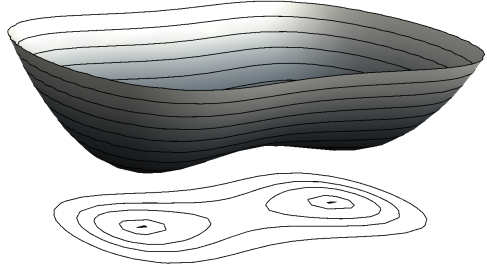}}
	\end{pspicture}
	\caption{Graph of the Hamiltonian function restricted to the transition manifold for the connecting example and its projections showing the transition states.}
	\label{ETMFig}
\end{figure}

As usual, we can write the transition states as level sets of the Hamiltonian function restricted to $N$ (see \figurename~\ref{ETMFig}),
$$H_N \left( u,v \right) = \frac{\alpha_2}{4} + \frac{\alpha_2}{2} \left(v^2 - u^2 \right) + \frac{\alpha_2}{4} u^4.$$
This has the origin, $\tilde{z}_0$, and $\tilde{z}_{\pm} =\left(\pm 1, 0 \right)$ as its critical points, with indices $\tilde{\lambda}_0=1$, $\tilde{\lambda}_\pm=0$. Thus the transition states bifurcate and change from $2\mathbb{S}^1$ to $1\mathbb{S}^1$. That is, as the energy is increased, the periodic orbits emanating from $\tilde{z}_\pm$ meet in homoclinic bifurcations at $\tilde{z}_0$ and connect to become one.
For energies below that at $\bar{z}_0$, this example therefore exhibits a transition state different from the usual basic scenario periodic orbit. This transition state is however the disjoint union of two periodic orbits, so if we had restricted our attention to a single index-1 critical point, we could have easily missed this more global picture.

Similarly, by considering $H_S$, we find that two dividing surfaces diffeomorphic to $\mathbb{S}^2$ about the index-1 critical points of $H$ connect to form a single sphere, $\mathbb{S}^2$, as the energy is increased.
The energy levels also bifurcate as we pass the critical point $\bar{z}_0$. Starting with an energy just above that at $\bar{z}_\pm$ and increasing it, we see that the two bottlenecks open up until they meet and become one, with the two regions of interest remaining the same.

We have considered here the uncoupled, symmetric case in which the three critical points are aligned on $x=0$ with $u \to -u$ symmetry. However, due to the persistence of normally hyperbolic submanifolds, adding coupling between the $(x,y)$ and $(u,v)$ degrees of freedom will not alter the conclusions about the bifurcation. For this, the normal hyperbolicity condition $\alpha_1 > \alpha_2$ is essential. We could also break the $u \to -u$ symmetry, in which case the saddles have different energy, so as the energy increases we first obtain one $\mathbb{S}^1$ then $2\mathbb{S}^1$, followed at the index-2 energy by qualitatively the same transition to $1\mathbb{S}^1$.

Now consider coupling our example to another oscillating degree of freedom. The Hamiltonian function for this could be
\begin{equation*}
H \left(z \right) = \frac{\alpha_2}{4} + \frac{\alpha_1}{2} \left(y^2 - x^2 \right) + \frac{\alpha_2}{2} \left(v^2 - u^2 \right) + \frac{\alpha_2}{4} u^4 + \beta \left(p^2 + q^2 \right) + \delta V \left( x, u, q \right) .
\end{equation*}
In the uncoupled case with $\delta = 0$, 
the transition state bifurcates from $2\mathbb{S}^3$ to $1\mathbb{S}^3$ and the dividing surface from $2\mathbb{S}^4$ to $1\mathbb{S}^4$.
In the coupled case, provided the coupling is sufficiently small, we can treat it as a perturbation of the uncoupled case and invoke the persistence of normally hyperbolic submanifolds to obtain topologically the same picture.


\subsection{Other Morse bifurcations}

If we restrict our attention to 2 degrees of freedom natural Hamiltonian systems, the critical points of the restricted Hamiltonian $H_N$ can only have index 0 or 1. 
At index-0 critical points a transition state is ``created'', whereas at critical energies corresponding to index-1 critical points, seeing as the transition state is closed, the transition state is generically a \textit{figure eight} (more complicated cases can occur if there are several critical points with the same energy).
Thus, the only generic bifurcations scenarios are the connection and disconnection ones found in Subsections \ref{volcano} and \ref{ezra}. 

It should be noted however, that this limitation on the types of Morse bifurcations of the transition states does not place significant restrictions on the bifurcations of the dividing surfaces. Just as we have found genus-1 dividing surfaces (as well as genus-0), we expect any genus surface should be possible. We expect that limitations will instead come from the transport problems and that these dividing surfaces will only appear in Hamiltonian systems for which the transport problem is not well defined. 

For natural systems with 3 degrees of freedom or higher, we have seen how the connecting and disconnecting scenarios with index-1 critical points of the restricted Hamiltonian function can be coupled to other degrees of freedom. Such systems may also have higher index critical points, which will give rise to other Morse bifurcations.
Explicit examples of connecting, disconnecting and also higher index Morse bifurcations will be seen in the next section in which a hypothetical class of planar bimolecular reactions is considered as an application of the previous sections.
Here we find various sequences of Morse bifurcations.


\section{Application. Planar bimolecular reactions}
\label{bimolecular}

As an application of our results, we shall consider (elementary) bimolecular reactions in gaseous phase,
$$A + B \rightarrow \text{ products},$$
and show how under a qualitative assumption on the interaction potential these can display interesting transition states and dividing surfaces, as a result of Morse bifurcations.
These have not been seen until now, as most studies have focused on the colinear and the zero angular momentum cases.
Here we shall consider the planar case, in which the Morse bifurcations stand out, and comment on the spatial case, leaving the detailed analysis to a subsequent paper. Some analysis of the structures in reaction dynamics in rotating molecules has been done recently in \cite{Ciftci2012}, but we are interested in the interaction of two rotating molecules.

Physically, we have a gas consisting of large numbers of $A$ and $B$ (polyatomic) molecules contained in some volume.
Reactions satisfying the usual assumptions of transition state theory can be described, using the Born-Oppenheimer approximation, as classical dynamics of the nuclei interacting via a potential given by the (ground state) energy of the electrons\footnote{Assumed non-degenerate and hence a smooth energy function, else it can have conical singularities, see e.g. \cite{Domcke2004}.} as a function of the internuclear coordinates, see e.g. Keck \cite{Keck1967}.
We thus have a very high dimensional Hamiltonian system representing a large number of interacting point particles.

Then, by assuming that the system is at any instant the product of ``reacting'' two molecule sub-systems that are independent of each other, we may consider an ensemble of these low dimensional Hamiltonian systems representing the reaction. 
That is we require the gas to be sufficiently dilute. 

The transport problem we consider is that of finding the rate of transport between the region of state space representing two distant molecules (reactants) and the region in which the molecules are close. The latter do not however generally constitute the products. Although the molecules must collide to react, whether the bonds are then broken and they go on to produce new molecules, with new bonds, is generally governed by further barriers. The rate obtained from this transport problem is therefore often referred to as the ``collision'' or ``capture'' rate. It provides a useful upper bound on the (equilibrium) reaction rate constant \cite{Chesnavich1980}. See also Keck \cite{Keck1967} and Henriksen and Hansen \cite[section 5.1]{Henriksen2008} for details of how to compute the rate constant from the flux (of state space volume).

Here, we shall always assume that the systems have Euclidean, i.e. translational and rotational, symmetry so that they can be reduced to a family of systems, in centre of mass frame, parametrised by angular momentum. 

The transport problem is clearly seen in the trivial example of a bimolecular reaction between two atoms (or charged particles).
Reducing the atom plus atom reacting system, 
we obtain a central force field Hamiltonian system $\left( T^* \mathbb{R}_+, \omega_0, H_\mu \right)$ with canonical sympletic form and
$$H_\mu \left(z \right) = \frac{1}{2m} p_r^2 + U_\mu \left( r \right), \quad U_\mu \left( r \right) = \frac{1}{2m}\frac{\mu^2}{r^2} + U \left( r \right),$$
where $m$ is the reduced mass, $\mu = \lvert L \rvert$ is the magnitude of the angular momentum, and $U_\mu$ is the \textit{effective potential} with the \textit{centrifugal term}.

The coordinate $r$ measures distance between the two atoms. We are interested in the case in which $U$ has a non-degenerate maximum at some $r = \bar R$, sufficiently large to keep us away from $r = 0$ which we are not interested in.
In this case, $U_\mu$ will also have a non-degenerate maximum at some $r = \bar R_\mu$, provided $\mu$ is not too large.
Note that this assumption is not necessary, as for certain attractive (away from $r = 0$) potentials $U$, the interplay between the potential and the centrifugal term will give a non-degenerate maximum of $U_\mu$, for some choices of $\mu$.
However, its hyperbolicity will depend on this interplay, so to keep the examples simple, especially in higher degrees of freedom, we assume that $U$ has a non-degenerate maximum.

The system has no centre directions at the critical point, seeing as there is only the radial degree of freedom. We can however still talk of transition and dividing manifolds in state space. Let the transition manifold be the critical point itself $N = \lbrace z \in M \vert p_r = 0, r = \bar R_\mu \rbrace$, so dimension-0.
A dividing manifold is then $S = \lbrace z \in M \vert r = \bar R_\mu \rbrace $ of dimension-1 and spanning $N$. The transition state is not defined for $E\neq E_a$, but the dividing surface for $E> E_a$ is $S_E = H_S (E) \cong \mathbb{S}^0$, i.e. two points.
For such potentials and with no other degrees of freedom, every barrier crossing will be reversed shortly after (except for $\mu=0$, if we do not consider repulsion).

The celestial two-body Kepler problem is nicely reviewed by Smale \cite{Smale1970}, as part of his ``topological program for mechanics'', in which he considers Hamiltonian systems with symmetries, the topology of their reduced state space and their reduced dynamics.

We shall now comment on the reduction procedures for molecular $n$-body problems.


\subsection{Molecular $n_a + n_b$ body problems}

The low-dimensional reacting system for two polyatomic molecules $A$ and $B$, with $n_a$ and $n_b$ atoms respectively, is given by $\left( T^* \mathbb R^{3(n_a + n_b)}, \omega_0 , H \right)$, where $\omega_0$ is the canonical symplectic form.
This is a natural Hamiltonian system on a cotangent bundle and has a Euclidean $SE(3) = \mathbb R^3 \otimes SO(3)$ group action on configuration space, whose lifted action on state space gives the system a symmetry, namely translational and rotational.

The translational symmetry is related to the conservation of the (total) linear momentum, via Noether's theorem. Reducing this symmetry gives the translation-invariant system $\left( T^* \mathbb R^{3(n_a + n_b -1)}, \omega_0 , H \right)$, crossed with the trivial centre of mass system that is ignored.

The rotational symmetry instead corresponds to an (lifted) $SO(3)$ action and is related to the conservation of angular momentum. This $SO(3)$ action is not free on the whole state space, so reduction gives a stratified reduced state space $M_\mu$ with symplectic strata of different dimensions.
Singular reduction theory, though quite involved, provides a global geometric picture.
Alternatively, seeing as the system is natural and has a lifted action, we can perform the reduction in configuration space instead by looking at the gauge theoretic cotangent bundle reduction picture, which is straightforward and gives an explicit choice of charts. 
This is nicely reviewed by Littlejohn and Reinsch \cite{Littlejohn1997}. 
The gauge in question is related to the Coriolis force.

First the translation symmetry is dealt with by choosing Jacobi vectors for the reacting system,
and forgetting about the trivial centre of mass system, thus obtaining the translation-invariant system. We choose to place a Jacobi vector along the line between the centres of the two molecules, the other vectors then describe the two molecules and can be chosen in a number of ways. See \figurename~\ref{jacobiFig} for the case of two reacting diatoms in the plane.
\begin{figure}
	\centering
	\def\JPicScale{0.5}
	\input{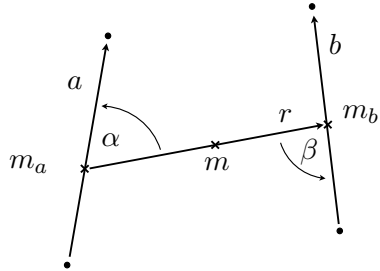}
	\caption{Choice of Jacobi vectors for two interacting diatoms in the plane, showing the vector along the line of centres between the molecules $r$ and the two other Jacobi vectors $a$ and $b$. The reduced masses are $m$, $m_{a}$, and $m_{b}$.}
	\label{jacobiFig}
\end{figure}
Then, we consider the separate cases for the $SO(3)$ action on configuration space, namely action on non-colinear, colinear and collisional configurations (of particles), and introduce a rotating frame. In the non-colinear case, e.g.~such that the line of centres is along the $x$-axis, $r = x$.
This corresponds to taking a section of $\tilde{Q}$, non-collinear configuration space, viewed as a fibre bundle with the orbits of the $SO(3)$ action as fibres, and quotient manifold $\tilde{\Sigma} = \tilde{Q}/SO(3)$, known as the \textit{internal space} (or \textit{shape space}), as a base space \cite{Littlejohn1997}. This gives a specific choice of coordinates.
Finally, seeing as the magnitude of the angular momentum is conserved, the angular momentum lives on a 2-sphere $\mathbb{S}^2_\mu$ and its dynamics is coupled to the internal motions, cf. the rigid body case.
One can complete the reduction by introducing canonical (Serret-Andoyer) coordinates on $\mathbb{S}^2_\mu$, but these introduce coordinate singularities and two charts are needed \cite{Deprit1967b}.

The state space of the reduced system is then a stratified fibre bundle,
and the dynamics can be considered separately on each stratum \cite{Iwai2005a}.
Considering the full dynamics across the strata is not straight forward.
Due to these complications, we consider planar systems that offer simple examples with which to study Morse bifurcations in bimolecular reactions without having to deal with technicalities of singularities and charts.
The interesting transition states and dividing surfaces are in no way limited to this case.


\subsection{Example 1. Planar atom plus diatom reactions}

The planar reduced three-body Hamiltonian for a bimolecular reaction between an atom and a diatom can be derived by choosing Jacobi vectors and then introducing a rotating frame, as explained above.
We choose a rotating frame such that the line of centres is along the $x$-axis, generally referred to as the $xxy$-gauge. The reduced Hamiltonian system is $\left( T^* \mathbb R^{3}_+, \omega_\mu , H_\mu \right)$ with
\begin{align*}
H_\mu \left(z \right) &= \frac{1}{2}\left(\frac{p_r^2}{m}+\frac{p_b^2}{m_b} + g^\beta p_\beta^2 \right) + U_\mu \left( r, b ,\beta \right) , \\
U_\mu \left( r, b ,\beta \right) &= \frac{1}{2} \frac{\mu^2}{I_3} + U \left( r, b ,\beta \right),\\
\omega_\mu &= \rmd q_i \wedge \rmd p_i 
- \mu \rmd \beta \wedge  \left( \del_{r} A_{\beta} \rmd r + \del_{b} A_{\beta} \rmd b \right),
\end{align*}
where $r$ is the distance from the atom to the centre of the diatom, $b$ is the length of the diatom, $\beta$ is its angle relative to the vector $r$, $\mu = \vert l \vert$ is the magnitude of the angular momentum, $m$ and $m_b$ are reduced masses, $ I_3 = m r^2+m_b b^2$ is the relevant part of the moment of inertia, $g^\beta = I_3/ \left( m m_b r^2 b^2 \right)$ is a component of the internal metric, and $A_\beta = m_b b^2 / I_3$ is the non-zero component of the gauge potential related to the Coriolis force. Note that we have chosen non-canonical coordinates in which the Coriolis effects come from the symplectic form. This simplifies the Hamiltonian function.

Assume that the potential $ U$ has a non-degenerate maximum with respect to $r$, at some $r = \bar{R} \left( b, \beta \right)$, which is sufficiently large and depends weakly on $b$, $\beta$. Then for $\mu$ small enough, the effective potential $U_\mu $ will have a non-degenerate maximum at some $r = \bar{R}_\mu \left( b, \beta \right)$, and depend weakly on $b, \beta$.
Note that for $\mu$ large, $U_\mu $ has no maximum any more, corresponding to distance of closest approach being too large. 

About the non-degenerate maximum in $r$, there is an almost invariant normally hyperbolic submanifold
$$N_0 = \lbrace z \in M_\mu \vert r = \bar{R}_\mu \left( b, \beta \right), p_r = 0 \rbrace.$$
This follows from the assumptions on the potential, namely that $\del_r U_\mu ( \bar{R}_\mu \left( b, \beta \right) , b, \beta ) = 0$, and that $\bar{R}_\mu $ is large so the vector field is almost tangent to $N_0$. Furthermore the assumption that the $b$ and $\beta$ dependence is weak implies that on $N_0$ the tangent dynamics is slower than the hyperbolic normal dynamics. 
Note that we are interested in the whole $(\beta, p_\beta)$ degree of freedom and do not necessarily require an index-1 critical point, though generically these will be present and we will consider the non-degenerate case with a critical point in the example Morse bifurcations below. 

We then choose as an approximate dividing manifold spanning $N_0$, the submanifold
$$S_0 = \lbrace z \in M_\mu \vert r = \bar{R}_\mu \left( b, \beta \right) \rbrace.$$

The Hamiltonian restricted to the approximate transition manifold $N_0$ is 
$$H_{N_0} \left( \tilde{z} \right) = \frac{1}{2} \left( \frac{p_{b}^2}{m_b}+ \bar{g}^\beta p_\beta^2 \right) +  \frac{1}{2} \frac{\mu^2}{\bar{I}_3} + \bar{U} \left( b ,\beta \right).$$

We now make explicit assumptions on the other degrees of freedom and therefore the potential, namely that as a function of $b$, $\bar{U}$ has a single non-degenerate minimum $\bar{b}(\beta)$.
We then consider three cases for the dependence of the potential on the angle $\beta$, see \figurename~\ref{atomDiatom}.
\begin{figure}
\begin{center}
\begin{pspicture}(14,4)
		\rput(2,2){\includegraphics[width=0.24\textwidth]{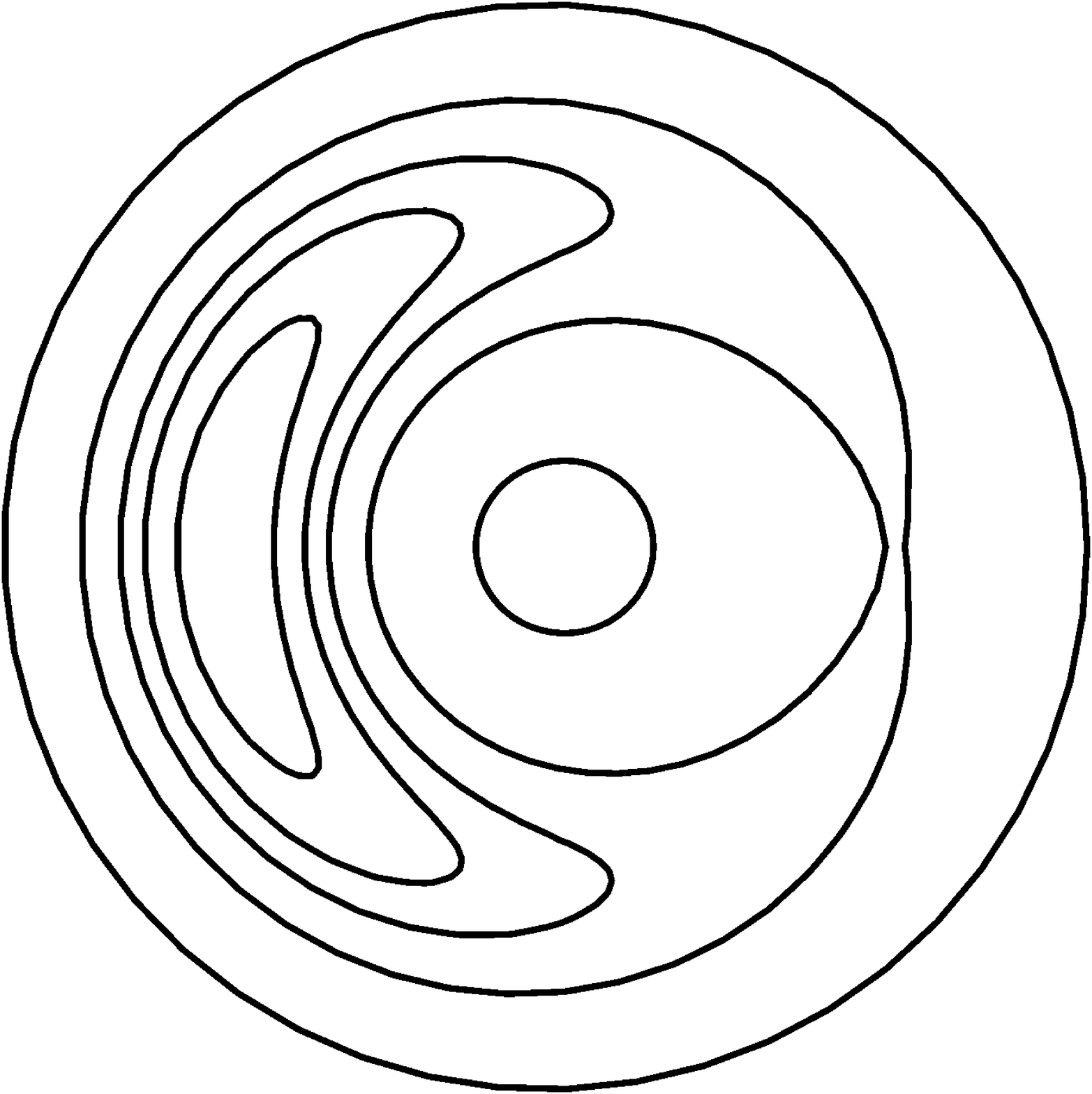}}
		\rput(7,2){\includegraphics[width=0.24\textwidth]{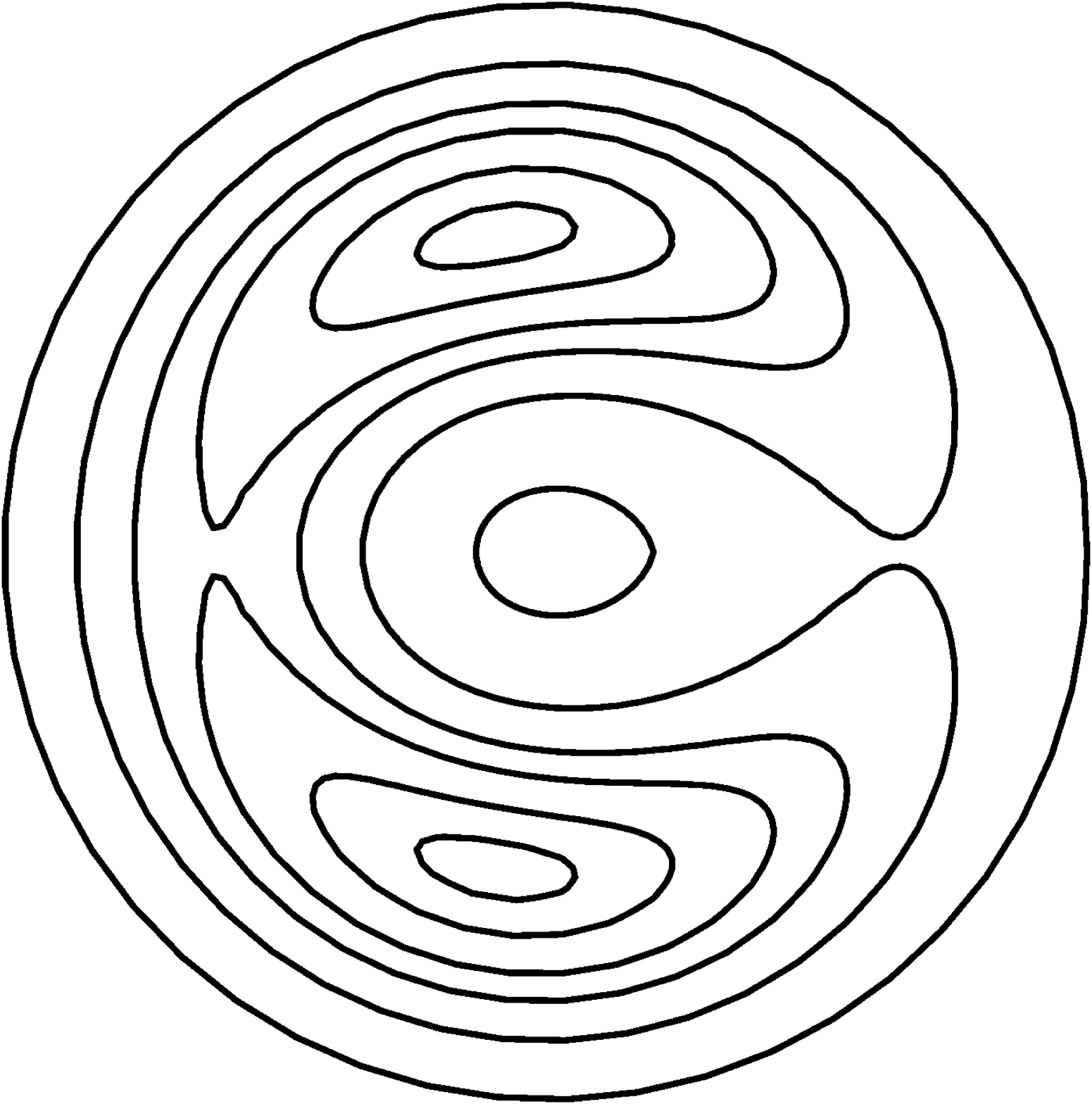}}
		\rput(12,2){\includegraphics[width=0.255\textwidth]{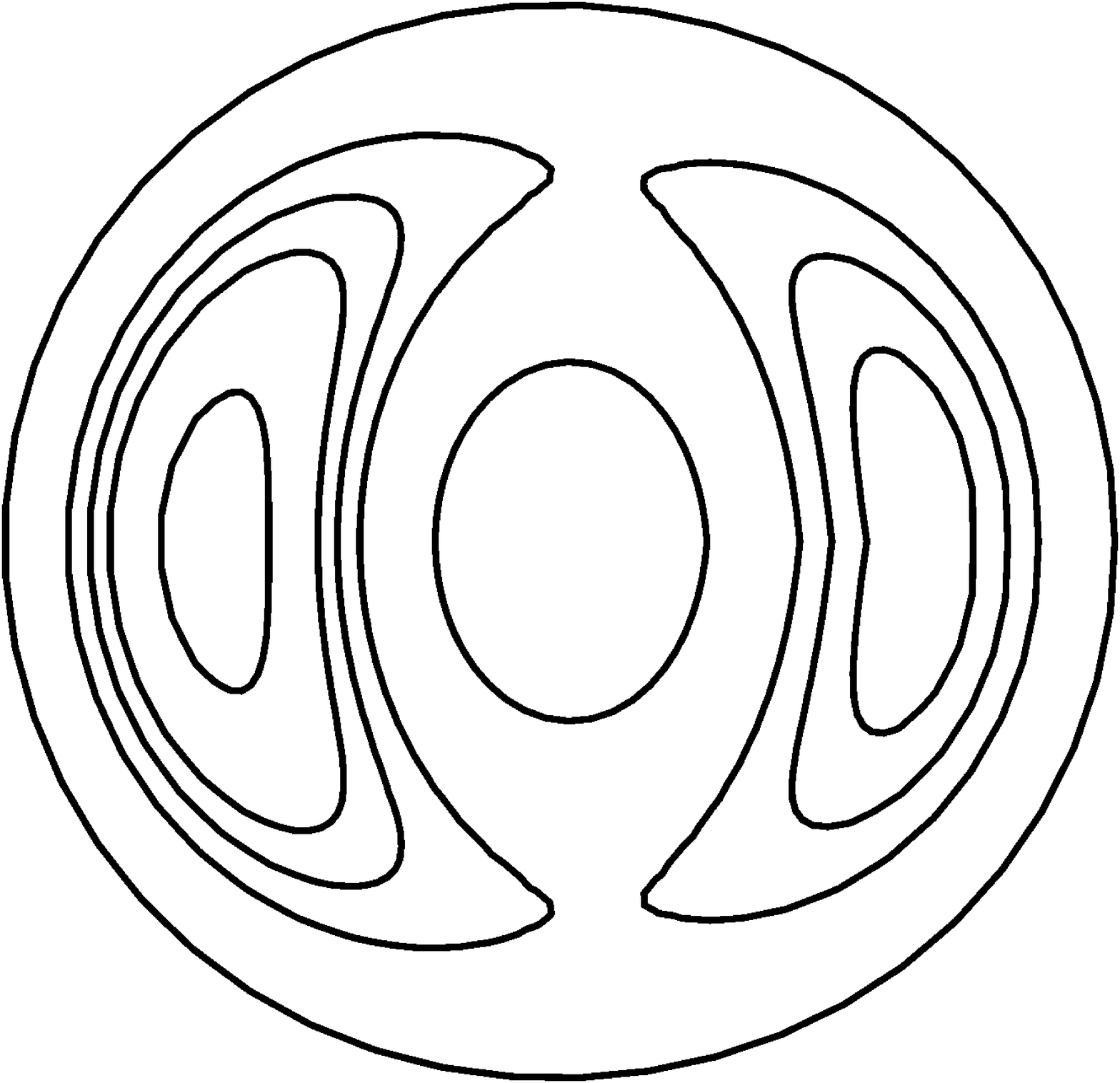}}
	\end{pspicture}
\end{center}
\caption{Contours of the restricted potential $\bar U$ for the atom plus diatom reaction in polar coordinates $(b,\beta)$. Left: Case 1. Atom attracted to one of the sides of the diatom (e.g. ion plus dipole). Centre: Case 2. Reaction prefers orthogonal configuration (e.g. atom plus non-symmetric non-polar diatom). Right: Case 3. Reaction prefers aligned configuration (e.g. atom plus non-symmetric dipole).}
\label{atomDiatom}
\end{figure}

To determine the Morse bifurcations, we can minimise the restricted Hamiltonian $H_{N_0}$ over the $(b,p_b)$ degree of freedom, by setting $b = \bar{b}(\beta)$, $p_{b}=0$, to obtain
$$\tilde{H}_{N_0} \left( \beta, p_\beta \right) = \frac{1}{2} \tilde{g}^\beta p_\beta^2 +  \frac{1}{2} \frac{\mu^2}{\tilde{I}_3} + \tilde{U} \left( \beta \right),$$
where $\tilde{U}(\beta) = \bar{U}(\bar{b}(\beta),\beta)$, 
and consider this simplified function. Keeping in mind the extra dimensions, this function gives the Morse bifurcations seeing as the other degree of freedom is not involved in the bifurcation, so the restricted Hamiltonian $H_{N_0}$ can be written as $\tilde{H}_{N_0}$ plus a remainder which is positive definite in $(b,p_b)$, by assumption.
The graph of the function $\tilde{H}_{N_0}$ looks the same as the potential on $N_0$, by replacing $b$ with $p_\beta$, see \figurename~\ref{atomDiatom}. 
The Morse bifurcations of the transition states and dividing manifolds are the following:
\begin{itemize}
\setlength{\itemindent}{15pt}
\item[Case 1.] This is a simple disconnecting bifurcation, see Section \ref{volcano}, with extra degrees of freedom. The transition state goes from $\mathbb{S}^3$ to $\mathbb{S}^1 \times \mathbb{S}^2$, and the dividing surfaces from $\mathbb{S}^4 $ to $ \mathbb{S}^1 \times \mathbb{S}^3$.
\item[Case 2.] Here the two minima of $\bar{U}$ are at the same height, whereas the saddles are at different heights; the transition state goes from $2 \mathbb{S}^3 $ to $ \mathbb{S}^3 $ to $ \mathbb{S}^1 \times \mathbb{S}^2$, and the dividing surface from $2 \mathbb{S}^4 $ to $ \mathbb{S}^4 $ to $ \mathbb{S}^1 \times \mathbb{S}^3$.
\item[Case 3.] Here the two minima of $\bar{U}$ are at different heights, whereas the saddles are at the same height; the transition state goes from $ \mathbb{S}^3 $ to $ 2 \mathbb{S}^3 $ to $ \mathbb{S}^1 \times \mathbb{S}^2$, and the dividing surface from $\mathbb{S}^4 $ to $ 2 \mathbb{S}^4 $ to $ \mathbb{S}^1 \times \mathbb{S}^3$.
\end{itemize}
Thus, we have found interesting sequences of connecting and disconnecting bifurcations in these bimolecular reactions.


\subsection{Example 2. Planar diatom plus diatom reactions}

The reduced reacting system for a planar diatom plus diatom bimolecular reaction, derived  as in the previous example, is $\left( T^* \mathbb R^{5}_+, \omega_\mu , H_\mu \right)$ with
\begin{align*}
H_\mu \left(z \right) &= \frac{1}{2}\left(\frac{p_r^2}{m} +\frac{p_a^2}{m_a} +\frac{p_b^2}{m_b} + \frac{p_\alpha^2}{m_a a^2} + \frac{\left( p_\alpha + p_\beta \right)^2}{m r^2} + \frac{p_\beta^2}{m_b b^2} \right) + U_\mu \left( r, a, \alpha, b ,\beta \right) , \\
U_\mu \left( r, a, \alpha, b ,\beta \right) &= \frac{1}{2} \frac{\mu^2}{I_3} + U \left( r, a, \alpha, b ,\beta \right),\\
\omega_\mu &= \rmd q_i \wedge \rmd p_i 
- \mu \left( \rmd \alpha \wedge  \del_{q_j} A_{\alpha} \rmd q_j + \rmd \beta \wedge \del_{q_k} A_{\beta} \rmd q_k \right),
\end{align*}
where $(a,\alpha)$ and $(b,\beta)$ describe the length and angle of the two diatoms, $\mu = \vert l \vert$ is the magnitude of the angular momentum, $m$, $m_a$ and $m_b$ are reduced masses, $ I_3 = m r^2 +m_a a^2 +m_b b^2$ is the relevant part of the moment of inertia, 
$A_\alpha = m_a a^2 / I_3$ and $A_\beta = m_b b^2 / I_3$ are the non-zero components of the gauge potentials.

The same assumptions as before give an approximate normally hyperbolic submanifold,
$$N_0 = \lbrace z \in M_\mu \vert r = \bar{R}_\mu \left( a, \alpha, b, \beta \right), p_r = 0 \rbrace,$$
where $R_\mu$ is the non-degenerate maximum of $U_\mu$ with respect to $r$, and is only weakly dependent on the other degrees of freedom.
Similarly, we choose the approximate dividing manifold $S_0= \lbrace r = \bar{R}_\mu \left( a, \alpha, b, \beta \right) \rbrace$.

Now, assuming that $U$ has a non-degenerate minimum with respect to $(a,b)$ for each $(r,\alpha,\beta)$,
we minimise over these degrees of freedom and consider the ``bifurcational'' restricted Hamiltonian
$$\tilde{H}_{N_0} \left( \alpha, \beta, p_\alpha, p_\beta \right) = \frac{1}{2} \left( \frac{p_\alpha^2}{m_a \bar a^2} + \frac{\left( p_\alpha + p_\beta \right)^2}{m \bar{R}_\mu^2} + \frac{p_\beta^2}{m_b \bar b^2} \right) +  \frac{1}{2} \frac{\mu^2}{\tilde{I}_3} + \tilde{U} \left( \alpha, \beta \right).$$

\begin{figure}
\begin{center}
\begin{pspicture}(12.5,5)
		\rput(2.5,2.5){\includegraphics[width=0.285\textwidth]{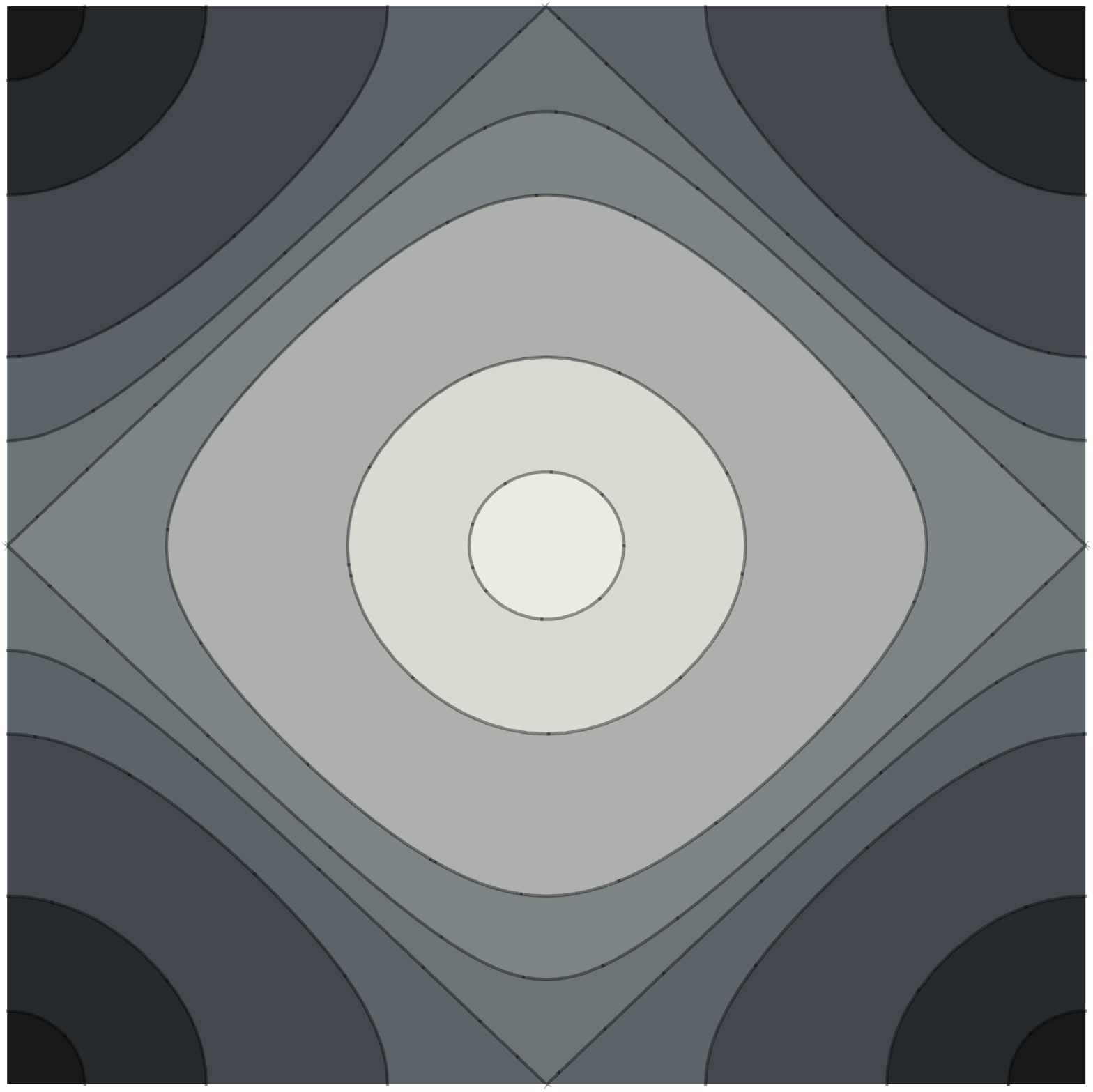}}
		\rput(9,2.5){\includegraphics[width=0.285\textwidth]{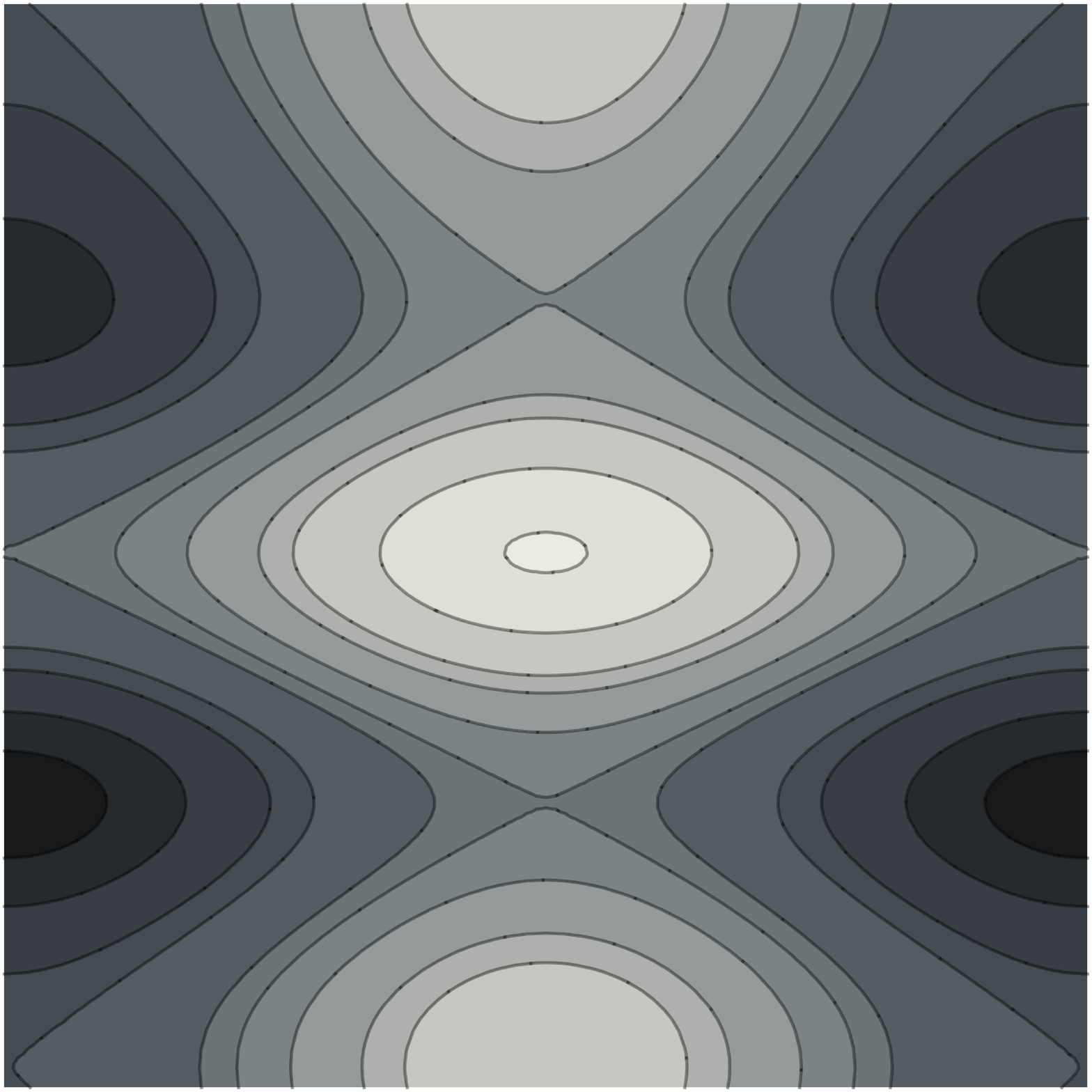}}
	\end{pspicture}
\end{center}
\caption{Contours of the frozen restricted potential $\tilde U$ for the diatom plus diatom reaction (darker shading means lower energy). Case 1. Simplest possible Morse function on $\mathbb T^2$ (assuming distinct saddles). Case 2. Possible restricted potential for dipole dipole reaction.}
\label{diatomDiatom}
\end{figure}

The behaviour depends on the qualitative form of the restricted potential $\tilde{U}$ on the 2-torus of $(\alpha,\beta)$. The Betti numbers of a torus $\mathbb{T}^2$ are 1,2,1. Thus, by the Morse inequalities, the simplest possible Morse function on $\mathbb{T}^2$ has four critical points of index 0,1,1,2. Assuming distinct saddle energies, contours of such a Morse function are given in \figurename~\ref{diatomDiatom}. Taking this as the restricted potential $\tilde{U}$ over $\mathbb{T}^2$, the resulting system would have a transition state bifurcating from $\mathbb{S}^7$ to $\mathbb{S}^1 \times \mathbb{S}^6$ to $\del X$ to $\mathbb{T}^2 \times \mathbb{S}^5$, where $\del X$ is the boundary of the handlebody $X = \left( \mathbb{S}^1 \times \mathbb{B}^7 \right) \cup_\psi \left( \mathbb{B}^1 \times \mathbb{B}^7 \right)$ given by Theorem \ref{thmB} and has no standard name. Similarly, the dividing surface would go from $\mathbb{S}^8$ to $\mathbb{S}^1 \times \mathbb{S}^7$ to $\del P$ to $\mathbb{T}^2 \times \mathbb{S}^6$, with $P$ again given by Theorem \ref{thmB}.

A more realistic frozen restricted potential $\tilde{U}$ for the case of two interacting dipoles will likely have more than 4 critical points thus leading to a longer sequence of Morse bifurcations. The contours of a possible example restricted potential with eight critical points at distinct heights are given in \figurename~\ref{diatomDiatom}. This has a sequence of critical points of index 0, 0, 1, 1, 1, 1, 2, 2. The transition state therefore goes from 
$$\mathbb{S}^7 \text{ to } 2 \mathbb{S}^7 \text{ to } \mathbb{S}^7 \text{ to } \mathbb{S}^1 \times \mathbb{S}^6 \text{ to } \del X \text{ to } \del Y \text{ to } \del Z \text{ to } \mathbb{T}^2 \times \mathbb{S}^5,$$
where again $\del X$, $\del Y$ and $\del Z$ are the boundaries of handlebodies given by Theorem \ref{thmB}. The dividing surface changes from
$$\mathbb{S}^8 \text{ to } 2 \mathbb{S}^8 \text{ to } \mathbb{S}^8 \text{ to } \mathbb{S}^1 \times \mathbb{S}^7 \text{ to } \del P \text{ to } \del Q \text{ to } \del R \text{ to } \mathbb{T}^2 \times \mathbb{S}^6,$$
again ending up diffeomorphic to $\mathbb{T}^2 \times \mathbb{S}^6$.


\section{Conclusions and discussion}
\label{conclude}

We have shown the existence of a class of systems for which the dividing surface method can be extended beyond the well known basic transport scenario by using Morse theory. 
This can be used to better understand many important transport problems, an example being bimolecular reaction, as seen from the hypothetical examples studied in Section \ref{bimolecular}. By considering the various different sequences of Morse bifurcations we were able to find interesting new transition states and dividing surfaces for the full non-colinear case, which has mostly been ignored until now. Here, we have considered the case of planar bimolecular reactions that provide simpler examples with which to focus on the Morse bifurcations. The spatial case, which displays similar but higher dimensional sequences of Morse bifurcations, reduces to an interesting Hamiltonian system with an additional angular momentum degree of freedom. This will be considered in a future publication.

Generally, we expect to see Morse bifurcations in a large class of transport problems from various applications.
In the context of applications, numerical methods to find and continue transition manifolds (through the Morse bifurcations) would be ideal. Some numerical methods for normally hyperbolic submanifolds do exist, however the high dimensionality of the transition manifolds and the global nature of the problem poses serious problems.

A natural question is how the flux of energy-surface volume through the dividing surface varies, as a function of energy, through a Morse bifurcation. In Appendix \ref{fluxCh}, we see that except in the 2 degree of freedom case the Morse bifurcations do not have a significant effect on the flux, which varies $C^{m-2}$ smoothly through these. In 2 degrees of freedom, the flux has a $-\Delta E \ln \vert\Delta E\vert$ infinite-slope singularity at an index-1 saddle on the transition manifold.

We have concentrated here on how the transition state and dividing surface vary with energy, but in a system depending smoothly on other parameters, a Morse bifurcation at energy $E_b$ for parameter value $\mu_b$ implies a Morse bifurcation at some nearby smoothly varying energy $E(\mu)$ for parameter $\mu$ near $\mu_b$. Thus for generic (i.e.~non-tangential) paths in the combined space $(E,\mu)$ there is a Morse bifurcation on crossing $ E= E(\mu)$.

Of the remaining open questions regarding the dividing surface method, one that is related to this work is  what happens when a dividing manifold cannot be defined over a large enough region of state space, such that even though one finds a sufficiently large normally hyperbolic potential transition manifold, which could be restricted to transition states, one cannot find dividing surfaces spanning them above certain energies? This breakdown of the dividing manifold, brought about by the intersection of the stable and unstable submanifolds of the transition manifolds, and its effect on transport problems needs to be investigated. For the case of 2 degrees of freedom, see e.g. Davis \cite{Davis1987}.

Finally there is the question of bifurcations leading to the loss of normal hyperbolicity of the transition manifold (for systems with more than 2 degrees of freedom) and their effect on Hamiltonian transport.
There have been some studies investigating the
loss of normal hyperbolicity for submanifolds of dimension greater than one, see e.g. \cite{li2009bifurcation, Teramoto2011, Allahem2012} and references therein, but it is still not understood. However, for Hamiltonian systems we expect that certain normally hyperbolic submanifolds will have extra geometric structure, namely symplecticity, and that this may provide an alternative way of understanding these complicated situations. See Conjecture \ref{claim} and the loss of symplectic nature of the transition manifold in the disconnecting example of Section \ref{volcano}.


\section*{Acknowledgements}
Invitation by Holger Waalkens and Arseni Goussev to attend a workshop at the University of Bristol in 2009 stimulated MacKay to point out the existence of the Morse bifurcations of transition states.
We would like to thank the Universit\'e Libre de Bruxelles for hospitality during part of the time in which this was written (Sep 2010 - Mar 2011), 
and Simone Gutt for sharing her insight on symplectic fibrations. 
Strub was supported by an EPSRC studentship.


\appendix


\section{Changes in flux of energy-surface volume at Morse bifurcations}
\label{fluxCh}

In order to find the rate of transport in a Hamiltonian system, once we have chosen a dividing surface $S_E$, we must find the flux of energy-surface volume through it in a given direction. In this appendix we address the shape of the flux as a function of energy. For example, we wish to connect to experiments such as in \cite{mackay1991a} (a different RS MacKay, we hasten to add!). Our approach is similar to that of Van Hove in his study of the singularities in the elastic frequency distribution of crystals \cite{VanHove}, which is related to the singularities of density of states, now known as Van Hove singularities.

Recall from Section \ref{divSurfs} that for regular energy levels $M_E$, the flux through a surface $S_E^+$ with boundary $N_E$ is
$$\phi_E \left( S_E^+ \right) = -\int_{N_E} \Lambda,$$
where $\Lambda = \frac{1}{\left( m - 1 \right) !} \lambda \wedge \omega^{m-2}$ is the ``generalised'' action form.

In the basic transport scenario, for small $\Delta E = E- E_a > 0$, the transition state $N_E = H_N^{-1} (E)$ is a small $(2m-3)$-sphere, so we can consider the leading orders of the restricted Hamiltonian $H_N$ in Williamson normal form. Then computing the integral, for example by passing to canonical action-angle variables, $J_i = \left( p_i^2 + q_i^2 \right)/2$ and $\theta_i = \arctan \left( p_i / q_i \right)$, and using Stokes theorem to integrate the volume of the ball $N_{\leq E}$ instead of the generalised action over the boundary $N_E$, gives the flux to leading order
$$\phi_E \left( S_E^+ \right) = \frac{\Delta E^{m-1}}{(m-1)!} \prod_{i=1}^{m-1} \frac{2 \pi}{\beta_i}.$$

Increasing the energy from that of the index-1 critical point $E_a$ of the basic scenario, we find that $N_E \cong \mathbb{S}^{2m-3}$ until the next critical value of $H_N$, by Theorem \ref{thmA}.

Away from critical values of $H_N$, the flux $\phi_E (S_E^+)$ as a function of the energy $E$ is $C^r$ if the restricted Hamiltonian is itself $C^r$. Note that even if $H$ is $C^\infty$, $H_N$ need not be very smooth, the most derivatives we can typically assume being given by the ratio of normal to tangential expansion at $N$. To see why the flux is $C^r$, use Stokes' theorem to write 
$$\phi_E (S_E^+)=-\int_{N_E} \Lambda = \int_{N_{\leq E}} V,$$
where $V=-\rmd \Lambda = \phi_E$ is a volume form on $N_{\leq E}$. Then consider an $e<E$ for which there are no critical values of $H_N$ in $[e,E]$ and rewrite the flux as
$$\phi_E (S_E^+)= \text{vol}(N_{\leq e}) + \int_{ H_N^{-1} ([e,E])} V.$$
Now, choose coordinates $z = (n,u,v) \in M$ so that $N_e = \{ z \in M \vert n=0,v=0 \}$ and $V = \rmd v \wedge V_e$ for some volume form $V_e$ on $N_e$. Then $N_E$ can be written as a graph over $N_e$, namely
$$N_E = \{ z \in M \vert H_N(z) = E \} = \{ z \in M \vert n=0,v=\bar{v}(u,E) \},$$
by the implicit function theorem, as $dH \neq 0$. The function $\bar{v}$ is $C^r$ if $H_N$ is $C^r$. Finally,
$$\phi_E (S_E^+) - \text{vol}(N_{\leq e}) 
= \int_{ N_e } \bar{v}(u,E) V_e$$
is $C^r$ in $E$. 

We shall therefore focus on flux for $E$ near a critical value $E_b$. The transition state cannot generally be defined in terms of local coordinates. However, by Theorem \ref{thmB} the sub-level sets just above a Morse bifurcation can be written as a handlebody, composed of a lower sub-level set and a handle, from which we deduce that interesting changes in the flux occur in a neighbourhood of the critical point, the contribution from the rest being $C^r$ and we will assume that $r$ is sufficiently large. Hence, we are concerned with changes in flux through Morse bifurcations and shall evaluate our integrals only over the handle region, see \figurename~\ref{thmBFig}.

The Morse lemma coordinates for a neighbourhood of a critical point $\bar{z}$ are not canonical so in general we do not know the form of the generalised action and cannot work with them. However, seeing as $N$ is symplectic, Darboux coordinates can be chosen in a neighbourhood of the critical point $\bar{z}$ in question such that the symplectic form is canonical.
Furthermore, the restricted Hamiltonian $H_N$ can be written in Williamson normal norm.

For an index-1 critical point $\bar{z}_2$ of $H_N$, the leading order terms of the Williamson normal form are
$$H_N \left(\tilde{z} \right) = E_b + \frac{\alpha_1}{2} \left(v_1^2 - u_1^2 \right) + \sum_{j=2}^{m-1} \frac{\beta_{j}}{2} \left(v_{j}^2 + u_{j}^2 \right).$$
Integrating the volume form $\phi_E = \omega^{m-1}/(m - 1)!$ over the sub-level set $N_{\leq E}$ restricted to a neighbourhood of $\bar{z}_2$ gives the contribution to the flux from this neighbourhood. Here we find a term of the form
\begin{align*} 
\vert\Delta E\vert^{m-1} &\ln \vert\Delta E\vert \text{ for } \Delta E < 0, \\
-\Delta E^{m-1} &\ln \vert\Delta E\vert \text{ for } \Delta E > 0,
\end{align*}
as well as terms with various powers of $\Delta E$. This term limits the smoothness of the flux as a function of the energy to $C^{m-2}$.

Similarly, we can consider how the flux changes at Morse bifurcations involving a critical point of any index. Ultimately, we are studying the differentiability of the volume of level sets about critical values. This has been studied by Hoveijn \cite{Hoveijn}, who tells us that the smoothness will always be limited to $C^{m-2}$, irrespective of the index. However, the nature of the discontinuity does depend on the index \cite[Prop.9]{Hoveijn}.

In conclusion, except when the number $m$ of degrees of freedom is small, Morse bifurcations do not have a significant effect on the flux of energy-surface volume, which varies $C^{m-2}$ smoothly through these. Provided $m-2 < r$, the Morse bifurcation will cause a small kink in the graph of $\phi_E (S_E^+)$ over $E$. We do not however expect that these will be visible from experimentally obtained reaction rates, in which other physical considerations probably have a larger impact on the shape of the graph.


\section{Approximating normally hyperbolic (symplectic) submanifolds of Hamiltonian systems}
\label{nhms}

Loosely speaking, a smooth, compact (possibly with boundary) invariant submanifold of a dynamical system is said to be normally hyperbolic if the linearised dynamics in the normal direction is hyperbolic and dominates the linearised tangent dynamics. 
Here, we shall recall a precise definition of a normally hyperbolic submanifold
and then present a method of finding approximations to normally hyperbolic symplectic submanifolds taken from MacKay's lectures on slow manifolds \cite{MacKay2004slow}.

We choose to work with the following
\begin{defn} Consider a dynamical system $\left( M , f^t \right) $ consisting of a $C^1$ flow $f^t$ on a smooth manifold $M$, and choose a Riemannian metric on $M$.
Let $N$ be a compact (possibly with boundary) $C^1$ submanifold of $M$ that is invariant under the flow, i.e. $f^t \left( N \right) = N$.
We say that $N$ is a \textit{normally hyperbolic submanifold} (of the dynamical system) if
the tangent bundle of $M$ restricted to $N$, $T_N M$, can be split continuously
$$T_N M = TN \oplus E^+ \oplus E^-,$$
such that $TN \oplus E^\pm$ are invariant under $D f^t$ for all $t$
and there exist real numbers $ k^\pm, k > 0$ and $ 0 \leq \beta < \alpha$,
such that for all $\tilde{z} \in N$, we have the following growth rates
\begin{align*}
\Vert \pi^+ \circ D_{\tilde{z}} f^t \vert_{E^+} \Vert &\leq k^+ e^{\alpha t} , \quad \forall t \leq 0, \\
\Vert \pi^- \circ D_{\tilde{z}} f^t \vert_{E^-} \Vert &\leq k^- e^{-\alpha t} , \quad \forall t \geq 0, \\
\Vert D_{\tilde{z}} f^t \vert_{TN} \Vert &\leq k e^{\beta \vert t \vert } , \quad \forall t \in \mathbb{R},
\end{align*}
where $\pi^\pm: T_N M \rightarrow E^\pm$ are the projections induced by the splitting.
\end{defn}

This is stronger than the usual definition, say that of Fenichel \cite[\S~IV]{Fenichel1971} as we can see from his uniformity lemma, but appropriate for our purposes. 

\begin{rmk}[Choice of splitting] Note that generally the normal hyperbolicity will depend on the choice of splitting.
There exists an invariant splitting that simplifies the theory, and is forced upon the definition in most of the literature. However, this choice of splitting is unnecessary and when considering concrete examples finding it can be cumbersome. Like Fenichel \cite{Fenichel1971}, which uses a Riemannian splitting, we choose a general splitting that is not invariant.
\end{rmk}

For the general properties of normally hyperbolic submanifolds, such as persistence, stable and unstable manifolds and smoothness results, see e.g. Fenichel \cite{Fenichel1971} or Hirsch, Pugh and Shub \cite{hirsch1977invariant}.

Given a vector field $X$ generating the flow $f^t$, the linearised flow $D f^t$ about the normally hyperbolic submanifold $N$ satisfies the (first) variation equation
$$ \frac{\rmd}{\rmd t} \left( D_{\tilde{z}} f^t \left( \nu \right) \right) = D_{f^t \left( \tilde{z} \right)} X \cdot D_{\tilde{z}} f^t \left( \nu \right) ,$$
for $\tilde{z} \in N$, $\nu \in T_{\tilde{z}} M$.
The splitting allows us to write $\nu = v_1 \xi + v_2 \eta_+ + v_3 \eta_-$. Let $v = \left( v_1,v_2, v_3 \right)$, and re-write the variation equation as 
$$\dot{v} = 
\left(
\begin{array}{ccc}
 T & C_+ & C_- \\
 0 & V_+ & 0 \\
 0 & 0 & V_-
\end{array}
\right)
v,$$
where we have used the invariance conditions.
Thus, by asking that
$$\Vert V_+^{-1} \Vert^{-1} \geq \alpha, \quad \Vert V_-^{-1} \Vert^{-1} \geq \alpha,  \quad \Vert T \Vert \leq \beta, \quad \Vert C_\pm \Vert \; \text{bounded},$$
we recover the conditions on the linearised flow from the definition. This way, for Hamiltonian systems we could ``re-write'' the definition in terms of properties of the linearised Hamiltonian (of the variation equation).

A main theorem on normally hyperbolic submanifolds tells us that given an ``almost invariant'' normally hyperbolic submanifold $N_0$ of a dynamical system, meaning that on $N_0$ the normal component of the vector field is small, there exists a true normally hyperbolic submanifold $N$ nearby.
Here, we are only interested in symplectic, normally hyperbolic submanifolds of Hamiltonian systems and we want to find sufficiently good approximations to the Hamiltonian on such a normally hyperbolic submanifold to deduce the sequence of Morse bifurcations as energy is increased. Thus, we will show how to find a better approximation of an almost invariant normally hyperbolic symplectic submanifold $N_0$, using a symplectically orthogonal fibration of its neighbourhood.
The approach follows MacKay's lectures that present the slow manifold case \cite{MacKay2004slow}. 

\begin{thm}
\label{approxNHM}
Every almost invariant, normally hyperbolic symplectic submanifold $N_0$ of a Hamiltonian system can be improved to one that also contains all nearby equilibria and has a smaller angle to the vector field.
\end{thm}

\begin{proof}
Given a Hamiltonian system $\left( M^{2m}, \omega, H \right)$ 
with vector field $X_H$ generating a flow $f^t$, and a 
symplectic submanifold $N_0$, consider a symplectic fibration of a { tubular} neighbourhood $U \subset M$ of $N_0$, $\pi: U \rightarrow N_0 : z \mapsto \tilde{z}$, { as defined in Section \ref{variation}}. The vertical subbundle is given by
$$\text{Vert}_z = \ker \rmd_z \pi = T_z F_{\tilde{z}}, \quad \forall z \in U,$$
and by choosing the horizontal subbundle to be symplectically orthogonal to the vertical subbundle, i.e.
$$ \text{Hor}_z = \text{Vert}_z^\omega,$$
we obtain a symplectic splitting of the tangent bundle in the neighbourhood
$$T M = \text{Vert} \oplus^\omega \text{Hor}.$$

{ Seeing as $N_0$ is a  symplectic submanifold of $M$, we can choose
(local) Darboux coordinates $z = \left(x,q,y,p \right)$ for a neighbourhood $U_0$ of $\tilde{z} \in N_0$ in which $\omega = \rmd x \wedge \rmd y + \rmd q \wedge \rmd p$ and such that $N_0 = \lbrace z \in M \vert x = y = 0 \rbrace$ and $F_{\tilde{z}} = \lbrace z \in M \vert q = \tilde q , \; p = \tilde p \rbrace$, where $\tilde z = (0, \tilde q, 0, \tilde p )$ and everything is restricted to $U_0$. We shall often write $h = (q,p)$ and $n = (x,y)$. To justify this local chart, we must note that we have restricted to a neighbourhood $V_0$ of $\tilde z$, such that $U_0 = \pi^{-1} (V_0)$ and by the local trivialisation is diffeomorphic to $V_0 \times F$. Then the chart is the symplectomorphism to $(\mathbb{R}^{2m},\omega_0)$ given by the symplectic neighbourhood theorem (see e.g.~\cite[Thm 3.30]{McDuff}). In these coordinates, the tangent space $\text{Hor}_z = \text{span} \lbrace \del_h \rbrace$ for $z \in F_{\tilde z}$, but globally the fibration may not be trivial and the Hor subbundle is not necessarily integrable (as discussed by Guillemin et al.~\cite[section 1.3]{Guillemin1996}).}
Now, we can write the { equations of motion} as
$${ \dot{z} = J \; DH \left( z \right) ,}$$
and the variation equation is 
$${ \dot v = J\;  D^2 H \left( f^t(\tilde{z}) \right) v. }$$

The assumptions of almost invariance, i.e. $\Vert D H \vert_{F_{\tilde{z}}} \left( \tilde{z} \right) \Vert \leq \varepsilon$ small for $\tilde{z} \in N_0$, and that the normal dynamics is hyperbolic, which can be written as $\Vert D^2 H \vert_{F_{\tilde{z}}} \left( \tilde{z} \right)^{-1} \Vert^{-1} \geq \alpha >0$, together with the implicit function theorem give the existence of a locally unique critical point { $ n_c ( \tilde{h} )$} of $H_{F_{\tilde{z}}}$ that is within approximately $\alpha^{-1} \varepsilon$ of $N_0$ and depends smoothly on { $\tilde{z} = (\tilde h, 0) \in N_0$}. Then define the new approximate submanifold $N_1$ to be the graph of $n_c$, so in particular $N_1$ contains all nearby true equilibria of the system. 
Note that finding $N_1$ does not require any special coordinates.
However, our choice of Darboux coordinates will now be used to
show that $N_1$ is a better type of approximation to the true normally hyperbolic submanifold $N$ than $N_0$, in the sense that the angle of the vector field to $N_1$ is small (called ``first order'' in \cite{MacKay2004slow}).  

Firstly, the restriction $\omega_{N_1}$ of $\omega$ to $N_1 $ is non-degenerate, and we use it to define $X_{H_{N_1}}$ tangent to $N_1$ via $\omega_{N_1} ( X_{H_{N_1}} , \zeta )= \rmd H_{N_1} ( \zeta)$ for all $\zeta \in T_z N_1$. Then, to check that $X_H - X_{H_{N_1}}$ is small compared to $X_{H_{N_1}}$, 
we first find that 
$\vert \symp{\zeta}{\eta} \vert \leq c \alpha^{-1} \delta \vert \zeta \vert \vert \eta \vert$ for $\zeta \in T_z N_1$, $\eta \in T_z F_{\tilde{z}}$ with $z \in N_1$, 
$\delta = \vert \del_{hn}^2 H (z) \vert$ and $c$ slightly larger than 1.
This can be seen by splitting the vectors tangent to $N_1$ into a horizontal and vertical part, namely $\zeta = \zeta_h \del_h + \zeta_n \del_n \in T_z N_1$, and writing $DH \vert_{F_{\tilde{z}}} (h, n_c)$ as $\del_n H (h, n_c)$. The tangent vectors satisfy
$$ \rmd (\del_n H) (\zeta) = \del^2_{hn} H (h, n_c) \zeta_h + \del^2_{nn} H (h, n_c) \zeta_n = 0,$$
so 
$$ \zeta_n = - (\del^2_{nn} H (h, n_c))^{-1} \del^2_{hn} H (h, n_c) \zeta_h = D_h n_c ( h) \zeta_h.$$
Then $\vert D_h n_c ( h) \vert \le c \alpha^{-1} \delta$ and $\vert \symp{\zeta}{\eta} \vert \leq c \alpha^{-1} \delta \vert \zeta \vert \vert \eta \vert$.
Next, we note that due to the definition of $N_1$, at $z \in N_1$ the vector field satisfies $\omega (X_H, \eta) =0$ and $\omega (X_H , \zeta ) = \omega (X_{H_{N_1}} , \zeta )$.
Finally, for a general $\nu \in T_z M$, split it as $\nu = \zeta + \eta$ with $\zeta \in T_z N_1$, $\eta \in T_z F_{\tilde{z}}$, then
\begin{align*}
\omega (X_H - X_{H_{N_1}} , \nu ) &= \omega ( X_H - X_{H_{N_1}} ,\zeta + \eta ) \\
&= \symp{X_H}{\zeta} + \symp{X_H}{\eta} - \omega ( X_{H_{N_1}} , \zeta) - \omega ( X_{H_{N_1}} , \eta ) \\
&= \mathcal{O} ( \alpha^{-1} \delta \vert X_{H_{N_1}} \vert \vert \eta \vert ).
\end{align*}
Thus $X_H - X_{H_{N_1}} = \mathcal{O} ( \alpha^{-1} \delta \vert X_{H_{N_1}} \vert )$, as claimed.
\end{proof}

Note however that beyond the first iteration, the required procedure is more subtle than \cite{MacKay2004slow} might lead one to suppose. { In order to ensure that for the successive approximations the normal vector field is of the order of higher powers of the tangential vector field, one has to carefully choose a nearly symplectically orthogonal fibration at each subsequent step.} MacKay gave a talk at the Newton Institute in Cambridge in 2007 { where this was addressed and an incomplete draft paper of 3 March 2007 sketches the procedure, but the paper has not yet been completed.}


\section{Morse theory}
\label{morse}

Morse theory allows us to study the topology of a manifold by considering the properties of ``height'' functions on it, and vice versa. It is therefore a natural tool for Hamiltonian systems with their Hamiltonian functions. 
We briefly state a few of the definitions and theorems (without proofs) and mention how they can be used to study bifurcations. For details see e.g. Milnor \cite{milnor1963morse} or Bott \cite{Bott1982}.

Consider an $m$-dimensional smooth manifold $M$ and a smooth function $H : M \rightarrow \mathbb{R}$. Recall that, a point $\bar{z} \in M$ is \textit{critical}, relative to $H$, if $\rmd_{\bar{z}} H = 0$. Given local coordinates $\left(x \right)$ about $\bar{z}$,  we have that
$$\frac{\del H}{\del x_1} \left( \bar{z} \right) = \cdots = \frac{\del H}{\del x_m} \left( \bar{z} \right) = 0.$$
Also, for a critical point $\bar{z}$, we can define a symmetric bilinear form, $\text{Hess}_{\bar{z}}\left(H\right)$, called the \textit{Hessian}. If $\xi, \eta$ are tangent vectors at $\bar{z}$, and $X,Y$ extensions to vector fields, we let $\text{Hess}_{\bar{z}}\left(H\right) \left( \xi, \eta \right) = X_{\bar{z}} \left( Y \left( H \right) \right)$. This is symmetric and independent of the extensions \cite[\S 2]{milnor1963morse}.

We can now give the
\begin{defn}
The \textit{(Morse) index} $\lambda \left(\bar{z} \right)$ of a critical point $\bar{z}$, relative to $H$, is the maximal dimension of a subspace $V$ of the tangent space on which the Hessian, $\text{Hess}_{\bar{z}}\left(H\right)$, is negative definite, that is $\text{Hess}_{\bar{z}}\left(H\right) \left(\xi, \eta \right)<0$ for all $\xi, \eta \in V$. The \textit{nullity}  of $\bar{z}$ relative to $H$ is the dimension of the null-space, i.e. the subspace consisting of all $\eta \in T_{\bar{z}} M$ such that $\text{Hess}_{\bar{z}}\left(H\right) \left(\eta, \xi \right)=0$ for all $\xi \in T_{\bar{z}} M$.
\end{defn}
In local coordinates, the index is the number of negative eigenvalues of the local representation of the  Hessian at $\bar{z}$, $ D^2 H \left(\bar{z} \right)$, counting multiplicities. The nullity is given by dim $M$ $-$ rank $D^2 H \left(\bar{z} \right)$.
Recall that a critical point is called \textit{nondegenerate} if the Hessian has nullity zero. 

Near a nondegenerate critical point, the level sets of $H$ are quadrics given by the
\begin{morseLem}
Let $\bar{z}$ be a nondegenerate critical point, relative to $H$, of index $\lambda$. Then, in some open neighbourhood of $\bar{z}$, there are local coordinates $\left(x,y \right)$ 
taking the critical point to the origin, and for which the local representation of $H$ satisfies
$$ H \left(x,y\right) = H \left(\bar{z}\right) - \frac{1}{2} \left( x_1^2 + \cdots + x_\lambda^2 \right) + \frac{1}{2} \left( y_1^2 + \cdots + y_{m-\lambda}^2 \right).$$
\end{morseLem}

Now, for a real number $a$, $M_{\leq a} = \lbrace z \in M \vert H\left(z\right) \leq a \rbrace$ is the \textit{sub-level set}, for $a$. If $a$ is a regular value of $H$, then these are manifolds with boundary $M_a = \partial M_{\leq a} = \lbrace z \in M \vert H\left(z\right)=a \rbrace $, the \textit{level sets}. Regarding these manifolds, we have
\begin{thm}
\label{thmA}
Let $a<b$ be real numbers with $H^{-1} \left([a,b]\right)$ compact.
Suppose $H^{-1} \left([a,b]\right)$ contains no critical points of $H$. Then $M_{\leq a}$ is diffeomorphic to $M_{\leq b}$, and hence so are their boundaries, $M_b \cong M_a$. 
Furthermore $H^{-1} \left([a,b]\right) \cong  M_a \times \left[ 0, 1 \right] \cong  M_b \times \left[ 0, 1 \right]$.
\end{thm}

To consider what happens when we ``pass'' a critical point, we need to recall how to attach handles. Firstly, we need the

\begin{defn}
An index-$\lambda$ \textit{handle}, or $\lambda$-handle of dimension $m$ is $h^m_\lambda = \mathbb{B}^\lambda \times \mathbb{B}^{m-\lambda}$, where $\mathbb{B}^k$ is the unit ball in $\mathbb{R}^k$. The \textit{axis} (also called \textit{core}) of the handle is $\mathbb{B}^\lambda \times \lbrace 0 \rbrace \subset h^m_\lambda$.
\end{defn}

Now, consider a manifold $M^m$ with boundary $\partial M$, a handle $h^m_\lambda$ and a smooth embedding $\psi : \mathbb{S}^{\lambda - 1} \times 	\mathbb{B}^{m-\lambda} \rightarrow \partial M$, called the \textit{attaching map}. We can form a topological space by taking the disjoint union, $ M \cup h^m_\lambda $, and then identifying $z$ in the boundary of the handle with $\psi \left( z \right) \in \partial M$. The quotient space thus obtained is denoted $ M \cup_\psi h^m_\lambda $.
Finally, we can show that $ M \cup_\psi h^m_\lambda $ admits a unique (up to diffeomorphism) smooth dimension-$m$ structure (see e.g. Milnor \cite[\S 3]{milnor1963morse}). 
Note that attaching a $0$-handle gives the disjoint union $ M \cup \mathbb{B}^m $.

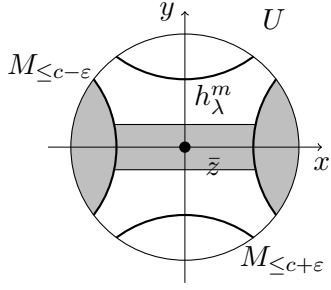
\begin{figure}
	\centering
	\usetikzlibrary{patterns}
\begin{tikzpicture}[ scale=1.5]

\begin{scope}  
\clip (0, 0) circle (1);

\draw[thick] (0,1.6) circle (1);
\draw[thick] (0,-1.6) circle (1);

\draw[fill=gray!50] (-1,0.2) -- (-0.5,0.2) --node[above]  {$h^{m}_\lambda$} (1,0.2) -- (1,-0.2)--(-1,-0.2)-- cycle;

\draw[thick,fill=gray!50] (1.6,0) circle (1);
\draw[thick,fill=gray!50] (-1.6,0) circle (1);
\end{scope}

\draw[black] (-1.7,0.4) node[label={60:$M_{\leq c-\varepsilon}$}] {};
\draw[black] (0.3,-1) node[label={right,above:$M_{\leq c+\varepsilon}$}]  {};

\draw[very thin,->] (-1.2,0) -- (1.2,0) node[below]  {$x$} ;
\draw[very thin,->] (0,-1.2) -- (0,1.2) node[left]  {$y$};

\draw (0,0) circle (1);
\draw[black] (0.5,0.9) node[label={30:$U$}]  {};

\fill (0,0) circle (0.05) node[label={-6:$\bar{z}$}] {};

\end{tikzpicture}
	\caption{Schematic representation of the handle attachment in theorem \ref{thmB}.}	\label{thmBFig}
\end{figure}
The effect of ``passing'' a critical point of the function on the diffeomorphism class of the sub-level sets is given by
\begin{thm}
\label{thmB}
Suppose $c$ is a critical value of $H$ such that $M_c$ contains a single nondegenerate critical point $\bar{z}$ of Morse index $\lambda$. Then for every $\varepsilon > 0$ sufficiently small, the sub-level set $M_{\leq c + \varepsilon}$ is diffeomorphic to $M_{\leq c - \varepsilon}$ with an index-$\lambda$ handle attached, i.e.
$$ M_{\leq c + \varepsilon} \cong M_{\leq c - \varepsilon} \cup_\psi h^m_\lambda .$$
\end{thm}
The handle attachment of Theorem \ref{thmB} is shown in \figurename~\ref{thmBFig}.
If we choose the Morse lemma coordinates $\left(x,y\right)$ about $\bar{z}$ the handle is given by
$$h^m_\lambda = \lbrace \lvert x \rvert^2 - \lvert y \rvert^2 \leq \varepsilon, \; \lvert y \rvert^2 \leq \delta \rbrace,$$
and the axis of the handle is given by
$$\mathbb{B}^\lambda \times \lbrace 0 \rbrace =  \lbrace \lvert x \rvert^2 \leq \varepsilon ,y = 0 \rbrace .$$
The figure also shows how the embedding is chosen naturally without ambiguity.

The representation of the sub-level set $M_{\leq c+\varepsilon}$ given in Theorem \ref{thmB} is that of a \textit{handlebody}, namely
\begin{defn}%
A manifold (with boundary in general) obtained from $\mathbb{B}^m$ by attaching handles of various indices one after another
$$\mathbb{B}^m \cup_{\psi_1} \mathbb{B}^{\lambda_1} \times \mathbb{B}^{m- \lambda_1} \cup_{\psi_2} \cdots \cup_{\psi_k} \mathbb{B}^{\lambda_k} \times \mathbb{B}^{m- \lambda_k}$$
is called an $m$-dimensional \textit{handlebody}.
\end{defn}
When considering bifurcations, we want to rewrite the diffeomorphism type in a more natural way. For this, we must know the topology of the sub-level set before the bifurcation, $M_{\leq c-\varepsilon}$, and the orientation of the handle with respect to the sub-level set. This orientation is given by the Morse coordinates. 

Note that, with a little care, we can also consider multiple (for a given critical value) and degenerate critical points.

Lastly, Theorem \ref{thmB} allows us to derive the Morse inequalities, which give bounds on the number of critical points and their indices based on the topology of the manifold. Firstly, we define the Morse series
$$\mathcal{M}_t \left( H \right) = \sum_{\bar{z}} t^{\lambda \left( \bar{z} \right)} ,$$
for critical points $\bar{z} \in C_r \left( H \right)$, then we need the Poincar\'e series
$$P_t \left( M \right) = \sum_k t^k b_k,$$
where $b_k = \dim H_k \left( M; \mathbb{R} \right)$ are the Betti numbers, i.e. the dimensions of the various homology groups of $M$ over the real numbers. These are topological invariants of the manifold, see e.g. Frankel \cite[Ch.13]{Frankel2004}. Finally, the Morse inequalities are
$$\mathcal{M}_t \left( H \right) - P_t \left( M \right) = \left( 1 + t \right) Q_t \left( H \right),$$
where $Q_t \left( H \right)$ is a polynomial in $t$ with non-negative coefficients, see e.g. Bott \cite{Bott1982}. One often writes the inequality $\mathcal{M}_t \left( H \right) \geq P_t \left( M \right)$ instead, hence the name. These relations can and have been used to study the ``potential energy surfaces'' of molecular dynamics, see Mezey \cite[Ch.2]{Mezey1987} and references therein.



\end{document}